\documentclass[12pt]{article}

\usepackage{a4wide}

\usepackage{xspace} 
\usepackage{eurosym} 
\usepackage{calc}
\usepackage{url}
\usepackage{color}
\usepackage{tikz}
\usetikzlibrary{arrows,patterns,plotmarks,er,3d,automata,backgrounds,topaths,trees,petri,mindmap}
\usepackage{graphicx}
\usepackage{array}
\usepackage[bookmarks={true},bookmarksopen={true}]{hyperref}
\usepackage[normalem]{ulem}

\newcommand{\pre}{\mathsf{pre}}

\usepackage{url}
\usepackage{latexsym,mathrsfs}
\usepackage{amsfonts, amsmath, amssymb}
\usepackage{graphics}

\usepackage{url}

\usepackage{changepage} 

\newcommand{\eq}{\leftrightarrow}
\newcommand{\Eq}{\Leftrightarrow}
\newcommand{\imp}{\rightarrow}
\newcommand{\Imp}{\Rightarrow}

\newcommand{\et}{\wedge}
\newcommand{\vel}{\vee}
\newcommand{\Et}{\bigwedge}
\newcommand{\Vel}{\bigvee}
\newcommand{\proves}{\vdash}

\newcommand{\is}{\exists}
\newcommand{\T}{\top}
\newcommand{\F}{\bot}
\newcommand{\Dia}{\lozenge}
\renewcommand{\Box}{\square}
\newcommand{\dia}[1]{\langle #1 \rangle}

\newcommand{\M}{M}   

\renewcommand{\phi}{\varphi}

\newcommand{\powerset}{{\mathcal P}}

\newcommand{\bisim}{{\raisebox{.3ex}[0mm][0mm]{\ensuremath{\medspace \underline{\! \leftrightarrow\!}\medspace}}}}

\newcommand{\bisrel}{\ensuremath{\mathfrak{R}}}

\newcommand{\lumis}{\succeq}

\newcommand{\weg}[1]{}

\newcommand{\Formulas}{{\mathcal L}}

\newcommand{\langu}{\Formulas}
\newcommand{\States}{S}

\newcommand{\state}{s}

\newcommand{\stateb}{t}

\newcommand{\Atoms}{P}

\newcommand{\atom}{p}

\newcommand{\Agents}{A}
\newcommand{\Group}{B}
\newcommand{\agent}{a}
\newcommand{\agenta}{\agent}
\newcommand{\agentb}{b}

\newcommand{\arel}{\ensuremath{\mathsf{R}}}
\newcommand{\Actions}{\ensuremath{\mathsf{S}}}
\newcommand{\actiona}{\ensuremath{\mathsf{s}}}

\newcommand{\Domain}{{\mathcal D}}
\newcommand{\dom}{\mathit{dom}}

\newcommand{\lang}{\langu}

\newcommand{\Lpl}{\ensuremath{\lang_{\mathit{pl}}}}

\newcommand{\langml}{\ensuremath{\lang_{\mathit{ml}}}}
\newcommand{\langpl}{\Lpl}
\newcommand{\langauml}{\ensuremath{\lang_{\mathit{auml}}}}

\renewcommand{\@}{{\color{red} @}}

\usepackage[thmmarks]{ntheorem}
\theorembodyfont{\rmfamily}
\theoremsymbol{\ensuremath{\dashv}}
\usepackage{newproof}

\newtheorem{theorem}{Theorem}

\newtheorem{definition}[theorem]{Definition}

\newtheorem{proposition}[theorem]{Proposition}

\newtheorem{corollary}[theorem]{Corollary}

\newtheorem{lemma}[theorem]{Lemma}

\newcommand{\aumodel}{U}
\newcommand{\austates}{O}
\newcommand{\austate}{o}

\newcommand{\aufunction}{R\!R}

\newcommand{\covers}{\nabla}

\newcommand{\AAUL}{\uparrow}

\begin{document}

\title{Arrow Update Synthesis}
\author{Hans van Ditmarsch\thanks{LORIA, CNRS, France, {\tt hans.van-ditmarsch@loria.fr}} $^,$\thanks{Hans van Ditmarsch is corresponding author. He is also affiliated to IMSc, Chennai, India. Support from ERC project EPS 313360 is kindly acknowledged.} \and Wiebe van der Hoek\thanks{Computer Science, University of Liverpool, United Kingdom, {\tt Wiebe.Van-Der-Hoek@liverpool.ac.uk}} \and Barteld Kooi\thanks{Department of Philosophy, University of Groningen, the Netherlands, {\tt b.p.kooi@rug.nl}} \and Louwe B.\ Kuijer\thanks{Computer Science, University of Liverpool, United Kingdom, {\tt Louwe.Kuijer@liverpool.ac.uk}}}
\date{}

\maketitle

\begin{abstract}
In this contribution we present {\em arbitrary arrow update model logic} (AAUML). This is a {\em dynamic epistemic logic} or {\em update logic}. In update logics, static/basic modalities are interpreted on a given relational model whereas dynamic/update modalities induce transformations (updates) of relational models. In AAUML the update modalities formalize the execution of {\em arrow update models}, and there is also a modality for quantification over arrow update models. Arrow update models are an alternative to the well-known action models. We provide an {\em axiomatization} of AAUML. The axiomatization is a rewrite system allowing to eliminate arrow update  modalities from any given formula, while preserving truth. Thus, AAUML is decidable and equally expressive as the base multi-agent modal logic. Our main result is to establish {\em arrow update synthesis}: if there is an arrow update model after which $\phi$, we can {\em construct} (synthesize) that model from $\phi$. 
We also point out some pregnant differences in {\em update expressivity} between arrow update logics, action model logics, and refinement modal logic. 
\end{abstract}

\paragraph*{Keywords:} modal logic, synthesis, dynamic epistemic logic, expressivity

\section{Introduction} \label{sec.introduction}

\paragraph*{Modal logic}
In modal logic we formalize that propositions may not be merely true or false, but that they are necessarily or possibly true or false, or that they may be desirable, or forbidden, or true later, or never, or that they are {\em known}. A common setting is for such modal propositions to be interpreted in \emph{relational models}, also known as \emph{Kripke models}. They consist of a domain of abstract objects, called states or worlds; then, given a set of labels, often representing agents, for each such agent a binary relation between those states; and, finally, a valuation of atomic propositions on the domain, typically seen as a unary relation, i.e., a property satisfied on a subset of the domain. The truth of a modal proposition is relative to a state in the relational model, called the actual state or the {\em point} of the model. The unit of interpretation is thus a {\em pointed model}: a pair consisting of a relational model and an actual state. 

If a pair $(s,s')$ is in the relation for $a$ this can mean that after executing action $a$ in state $s$ the resulting state is $s'$. But it can also mean that agent $a$ considers state $s'$ desirable in case she is in state $s$. The interpretation that we focus on, is that of \emph{information}. That is, it is consistent with $a$'s information in state $s$ that the state would be $s'$. In state $s$ it is true that agent $a$ \emph{knows} $\phi$ (or \emph{believes} $\phi$, depending on the properties of the relation), notation $\Box_a \phi$, if the formula $\phi$ is true in all states $s'$ accessible from $s$, i.e., for all $s'$ with $(s,s')$ in the  relation for $a$. The modal logics using that kind of interpretation of modalities are called epistemic logics  \cite{modalhandbook,hvdetal.handbook:2015}.

As an example, consider two agents $a,b$ (commonly known to be) uncertain about the truth of a propositional variable $p$. The uncertainty of $a$ and $b$ can be pictured as follows. We `name' the states with the value of the variable $p$. The actual state is framed. Pairs in the accessibility relation are visualized as labelled {\em arrows}. In the actual state: $p$ is true, agent $a$ does not know $p$ because she considers a state possible wherein $p$ is false (formally $\neg \Box_a p$), agent $a$ also does not know $\neg p$ because she considers a state possible wherein $p$ is true (formally $\neg \Box_a \neg p$, also written as $\Dia_a p$), and similarly for agent $b$. Agent $a$ also knows that she is ignorant about $p$, as this is true in both states that she considers possible. The accessibility relations for $a$ and $b$ are both equivalence relations. This is always the case if the modalities represent knowledge. 

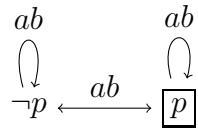
\begin{center}
\begin{tikzpicture}
\node (0) at (0,0) {$\neg p$};
\node (1) at (2,0) {$\fbox{$p$}$};
\draw[<->] (0) -- node[above] {$ab$} (1);
\draw[->] (0) edge[loop above,looseness=15] node[above] {$ab$} (0); 
\draw[->] (1) edge[loop above,looseness=9] node[above] {$ab$} (1); 
\node (1t) at (7,0) {\color{white}$\fbox{$p$}$};
\end{tikzpicture}
\end{center}

\paragraph*{Update logic} In this work we focus on modal logics that are \emph{update logics}. Apart from the modalities that are interpreted {\em in} a relational model, they have other modalities that are interpreted by {\em transforming} a relational model (and by then interpreting the formulas bound by that modality in the transformed model). If the modal logic is an epistemic logic, update logics are called {\em dynamic epistemic logics}. To distinguish them we call the former {\em static} and we call the latter {\em dynamic}.

The {\em updates} $X$ that we consider can be defined as transformers of relational models. This transformation induces a binary update relation between pointed models. To an update relation corresponds an {\em update modality} (often also called update) that is interpreted with this relation, so we can see those as $[X]$ or $\dia{X}$, where $[X]\phi$ means that $\phi$ is true in all pointed models transformed according to the $X$ relation, and  $\dia{X}\phi$ that there is a pair of pointed models in the relation. Given a relational model we can change its domain of states, the relations between the states, or the valuations of atomic propositions, or two or more of those at the same time. There are therefore many options for change. Change the valuation of a model is also known as {\em factual change} \cite{hvdetal.aamas:2005,jfaketal.lcc:2006}. Update involving factual change is an interesting topic, but it is outside the scope of the current paper.

\paragraph*{Public announcement logic}
The basic update for states is the model restriction, and the basic update operation interpreted as a model restriction is a \emph{public announcement}. The logic with epistemic modalities and public announcements is {\em public announcement logic} (PAL) \cite{plaza:1989,baltagetal:1998}. A public announcement of $\phi$ restricts the domain to all states where the announced formula $\phi$ is true, thereby decreasing the uncertainty of the agents. As a result of the domain restriction, the relations and the valuation are adjusted in the obvious way. A condition for the transformation is that the actual state is in the restriction. This means that the announcement formula is true when announced.

As an example, after the public announcement of $p$, both $a$ and $b$ know that $p$:

\begin{center}
\begin{tikzpicture}
\node (0) at (0,0) {$\neg p$};
\node (1) at (2,0) {$\fbox{$p$}$};
\draw[<->] (0) -- node[above] {$ab$} (1);
\draw[->] (0) edge[loop above,looseness=15] node[above] {$ab$} (0); 
\draw[->] (1) edge[loop above,looseness=9] node[above] {$ab$} (1); 
\node (t) at (3.5,0) {\large $\Imp$};
\node (1t) at (7,0) {$\fbox{$p$}$};
\draw[->] (1t) edge[loop above,looseness=9] node[above] {$ab$} (1t); 
\end{tikzpicture}
\end{center}

\paragraph*{Arrow update logic}
The basic update for relations is the relational restriction, i.e., a restriction of the {\em arrows}: a pair in the relation is called an `arrow'. This leaves all states intact, although some may have become unreachable. In {\em arrow update logic} (AUL), proposed in \cite{kooietal:2011} we specify which arrows we wish to preserve, by way of specifiying what formulas should be satisfied at the source (state) of the arrow and the target (state) of the arrow. This determines the model transformation. Such a specification is called an {\em arrow update}. The logic AUL contains modalities for arrow updates.

Given initial uncertainty about $p$ with both agents, a typical arrow update is the action wherein Anne ($a$) opens an envelope containing the truth about $p$ while Bill ($b$) observes Anne reading the contents of the letter. We preserve all arrrows satisfying one of $p \imp_a p, \neg p \imp_a \neg p$, and $\T \imp_b \T$. Therefore, only two arrows disappear, $\neg p \imp_a p$ and $p \imp_a \neg p$.

\begin{center}
\begin{tikzpicture}
\node (0) at (0,0) {$\neg p$};
\node (1) at (2,0) {$\fbox{$p$}$};
\draw[<->] (0) -- node[above] {$ab$} (1);
\draw[->] (0) edge[loop above,looseness=15] node[above] {$ab$} (0); 
\draw[->] (1) edge[loop above,looseness=9] node[above] {$ab$} (1); 
\node (t) at (3.5,0) {\large $\Imp$};
\node (0t) at (5,0) {$\neg p$};
\node (1t) at (7,0) {$\fbox{$p$}$};
\draw[<->] (0t) -- node[above] {$b$} (1t);
\draw[->] (0t) edge[loop above,looseness=15] node[above] {$ab$} (0t); 
\draw[->] (1t) edge[loop above,looseness=9] node[above] {$ab$} (1t); 
\end{tikzpicture}
\end{center}

The boundary between state elimination and arrow elimination is subtle. If $p$ is true, the following arrow update with $\T \imp_a p, \T \imp_b p$ is the same update as a public announcement of $p$. This is because there is no arrow from the $p$ state to the $\neg p$ state after the update. Therefore, if $p$ is true, the $\neg p$ state does not matter.
In another formalism this arrow update is known as the arrow elimination semantics of public announcement \cite{gerbrandyetal:1997,kooi.jancl:2007}. 

\begin{center}
\begin{tikzpicture}
\node (0) at (0,0) {$\neg p$};
\node (1) at (2,0) {$\fbox{$p$}$};
\draw[<->] (0) -- node[above] {$ab$} (1);
\draw[->] (0) edge[loop above,looseness=15] node[above] {$ab$} (0); 
\draw[->] (1) edge[loop above,looseness=9] node[above] {$ab$} (1); 
\node (t) at (3.5,0) {\large $\Imp$};
\node (0t) at (5,0) {$\neg p$};
\node (1t) at (7,0) {$\fbox{$p$}$};
\draw[->] (0t) -- node[above] {$b$} (1t);
\draw[->] (1t) edge[loop above,looseness=9] node[above] {$ab$} (1t); 
\end{tikzpicture}
\end{center}

\paragraph*{Generalizations} In PAL and AUL the complexity (the number of states) of the relational model cannot increase. By generalizing the mechanism underlying state elimination and arrow elimination we can achieve that, and thus express more model transformations. This increases their {\em update expressivity}. From the perspective of information change, this adds uncertainty about what is happening. We obtain \emph{action models} \cite{baltagetal:1998} as a generalization of public announcements, and \emph{arrow update models} \cite{kooirenne} as a generalization of arrow updates. 

\paragraph*{Action model logic}
{\em Action model logic} (AML) was proposed by Baltag, Moss and Solecki in \cite{baltagetal:1998}. An {\em action model} is like a relational model but the elements of the domain are called {\em actions} instead of states, and instead of a valuation a {\em precondition} is assigned to each domain element. The transformed relational model is then the modal product of the relational model and the action model, restricted to (state,action) pairs where the action can be executed in that state. We refer to Section \ref{sec.arrowvaction} for a formal introduction. 

An example is the action as above wherein Anne reads the contents of a letter containing $p$ or $\neg p$, but now with the increasing uncertainty that Bill is uncertain whether Anne has read the letter (and that they are both aware of these circumstances). The action model is not depicted (details are in Section \ref{sec.arrowvaction}). The model transformation is as follows. In the resulting framed state, $a$ knows that $p$, but $b$ considers it possible that $a$ is uncertain about $p$, i.e., $\Box_a p \et \Dia_b \neg (\Box_a p \vel \Box_a \neg p)$. In the figure we assume transitivity of the relation for $b$.

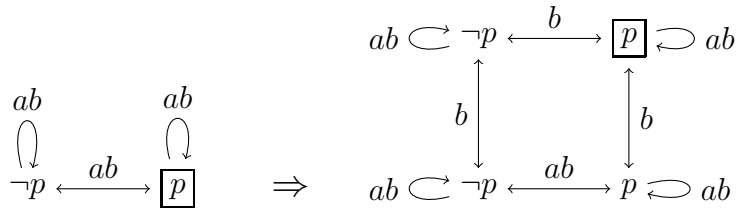
\begin{center}
\begin{tikzpicture}
\node (m) at (-2.6,0) {\color{white}$\neg p$};
\node (0) at (0,0) {$\neg p$};
\node (1) at (2,0) {$\fbox{$p$}$};
\draw[<->] (0) -- node[above] {$ab$} (1);
\draw[->] (0) edge[loop above,looseness=15] node[above] {$ab$} (0); 
\draw[->] (1) edge[loop above,looseness=9] node[above] {$ab$} (1); 
\node (t) at (3.5,0) {\large $\Imp$};
\node (0r) at (6,0) {$\neg p$};
\node (1r) at (8,0) {$p$};
\node (0ra) at (6,2) {$\neg p$};
\node (1ra) at (8,2) {$\fbox{$p$}$};
\draw[<->] (0r) -- node[above] {$ab$} (1r);
\draw[<->] (0ra) -- node[above] {$b$} (1ra);
\draw[<->] (0r) -- node[left] {$b$} (0ra);
\draw[<->] (1r) -- node[right] {$b$} (1ra);
\draw[->] (0r) edge[loop left,looseness=9] node[left] {$ab$} (0r); 
\draw[->] (1r) edge[loop right,looseness=15] node[right] {$ab$} (1r); 
\draw[->] (0ra) edge[loop left,looseness=9] node[left] {$ab$} (0ra); 
\draw[->] (1ra) edge[loop right,looseness=9] node[right] {$ab$} (1ra); 
\end{tikzpicture}
\end{center}

\noindent Similar logics (or semantics) for action composition are found in \cite{hvd.thesis:2000,kooi:2003,jfaketal.lcc:2006,aucher:2010,EijckSW11}. Action model logic is often referred to as (the) dynamic epistemic logic. As said, we use the latter more generally, namely to denote any update logic with an epistemic interpretation.

\paragraph*{Arrow update model logic} {\em Generalized arrow update logic} \cite{kooirenne} is a (indeed) generalization of arrow update logic where the dynamic modalities for information change formalize execution of (pointed) {\em arrow update models}, structures akin to the {\em action models} of action model logic. In this contribution, instead of generalized arrow update logic we call it {\em arrow update model logic} (AUML). The arrow updates of \cite{kooietal:2011} correspond to singleton arrow update models. The next Section \ref{sec.aauml} formally introduces them. The above is also an example of arrow update model execution --- Section \ref{sec.arrowvaction} explains how to get action models from arrow update models and vice versa, and to what extent they define the same update.
   
\paragraph*{Quantification over information change}   
Another extension of update logics is with quantification over updates. {\em Arbitrary public announcement logic} (APAL) adds quantification over public announcements to PAL \cite{balbianietal:2008}. {\em Arbitrary arrow update logic} (AAUL) \cite{hvdetal.aij:2017} extends arrow update logic with quantifiers over information change induced by arrow updates: it contains dynamic modalities formalizing that {\em there is an arrow update after which} $\phi$. {\em Arbitary action model logic} (AAML) by Hales \cite{hales2013arbitrary} add quantifiers over action models to AML. In {\em arbitrary arrow update model logic} (AAUML), the topic of this paper, we add quantifiers over arrow update models to the logical language. It is like Hales' arbitrary action model logic. {\em Refinement modal logic} (RML) \cite{bozzellietal.inf:2014} has a modality representing quantification over updates, but does not have (deterministic/concrete) update modalities in the object language to quantify over. We show that the AAML and AAUML quantifier behave much (but not quite) like the refinement quantifier in RML. Section \ref{sec.rml} is devoted to it.

Figure~\ref{fig:logic_relations} gives an overview of the different logics discussed in the paper, in their relation to AAUML. The four logics in the left square are based on state manipulation, the four logics in the right square are based on arrow manipulation. Entirely on the left we find the base modal logic ML and the logic RML, that is also arrow manipulating.

All these logics are equally expressive as ML and are decidable, which can be shown by truth-preserving rewriting procedures to eliminate updates (for AAUML this is one of the results of the paper), except for APAL and AAUL, which are more expressive and undecidable \cite{balbianietal:2008,frenchetal:2008,hvdetal.aij:2017,hvdetal.undecidable:2017}.   However, the logics greatly differ in update expressivity, as the typical examples above already demonstrated. See also Sections \ref{sec.updateex} --- \ref{sec.rml}. Finally, it should be mentioned that all the logics are invariant under bisimulation. This is because the parameters of the model transforming dynamic modalities and quantifiers are (model restrictions induced by) formulas.

There are many other updates and update logics that we do not consider in this paper. In particular we do not consider updates $X$ that can only be defined as {\em pointed} model transformers (that is, they cannot be {\em globally} defined on the entire model; they are defined {\em locally}: how they transform the model depends on the actual state). If such were the definition of an update, even the interpretation of a static modality can be seen as an update, namely transforming the model with point $s$ into the model with point $s'$, where the point has shifted given a pair $(s,s')$ in the relation for an agent. Such local update logics are often more expressive than modal logic, are often undecidable, are typically not invariant under (standard) bisimulation, and may lack axiomatizations. Examples are \cite{jfak.sabotage:2005,aucheretal:2009,arecesetal:2012}. In \cite{arecesetal:2012} not only relational {\em restriction} is considered but also relational {\em expansion} (`bridge') and relational {\em change} that is neither restriction nor expansion, such as reversing the direction of arrows (`swap'). It should finally be noted that the distinction between static modalities, interpreted in a model, and dynamic modalities, interpreted as updates, is not rigid: unifying perspectives include \cite{jfaketal.lcc:2006}.

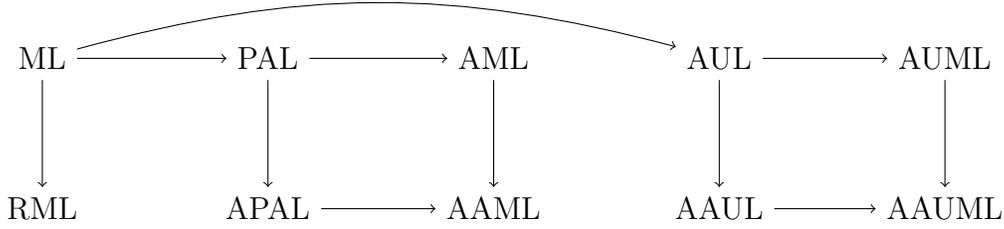
\begin{figure}
\begin{center}
\begin{tikzpicture}
\node (m33) at (-3,2) {ML};
\node (m30) at (-3,0) {RML};
\draw[->](m33) -- (m30);
\node (03) at (0,2) {PAL};
\node (00) at (0,0) {APAL};
\node (33) at (3,2) {AML};
\node (30) at (3,0) {AAML};
\draw[->] (03) --  (33);
 \draw[->] (00) -- (30);
 \draw[->] (03) -- (00);
 \draw[->] (33) -- (30);
 
\node (63) at (6,2) {AUL};
\node (60) at (6,0) {AAUL};
\node (93) at (9,2) {AUML};
\node (90) at (9,0) {AAUML};
 
 \draw[->] (63) -- (93);
 \draw[->] (60) -- (90);
\draw[->] (63) -- (60);
\draw[->] (93) -- (90);

\draw[->](m33) --  (03);
\draw[->,bend left=15] (m33) to (63);
\end{tikzpicture}
\end{center}
\caption{An overview of update logics discussed in the paper. Horizontal arrows informally represent {\em more complex} updates. Vertical arrows informally represent {\em quantification over} updates. The arrows can be interpreted as syntactic extensions (modulo the names of quantifiers) or as semantic generalizations. Assume transitive closure.}
\label{fig:logic_relations}
\end{figure}

\paragraph*{Synthesis} 
For these update logics we can ask whether there is an update that achieves a certain goal. For the logics without quantification this question cannot be asked in the object language but only meta-logically. That is, we can ask whether there is an update $X$ such that $\langle X\rangle \phi$ is true. For the update logics with quantification this question can be asked in the object language. Let $\dia{?}$ be (the existential version of) that quantifier. Then $\langle ?\rangle \phi$ asks whether there is an update $X$ that makes $\phi$ true. 

Only knowing \emph{whether} there is an update that achieves a goal is not very satisfying; we would also like to know \emph{which} update, if any, achieves the goal. So we would like to know not only whether the goal is achievable but also how it can be achieved. The process of constructing this update that achieves the goal is known as \emph{synthesis}.

Formally, the synthesis problem for a given type of update takes as input a formula $\phi$, and gives as output an update $X$ of that type such that, whenever $\phi$ can be achieved, then $X$ achieves $\phi$. In symbols, this is the validity of $\langle ?\rangle \phi \rightarrow \langle X \rangle \phi$.

This is a rather strong goal. We do not consider it sufficient to find, for every pointed model $(M,s)$, an update $X_{(M,s)}$ such that $(M,s)$ satisfies $\langle?\rangle\phi \rightarrow \langle X_{(M,s)}\rangle \phi$. We want one single update $X_\phi$ that achieves $\phi$ in every model where $\phi$ is achievable.
Because this goal is so strong, there is, in general, no guarantee that synthesis is possible. 

For PAL this strong kind of synthesis is impossible. If $(M_1,s_1)$ satisfies $\langle\psi_1\rangle\phi$ and $(M_2,s_2)$ satisfies $\langle\psi_2\rangle\phi$, so if $\phi$ can be achieved in two different situations using two different public announcements, then there is typically no unifying public announcement $\psi$ such that $(M_1,s_2)$ satisfies $\langle\psi\rangle \phi$ and $(M_2,s_2)$ satisfies $\langle\psi\rangle\phi$.\footnote{For example, consider the four-state model below; $\overline{p}$ means that $p$ is false in that state, etc. Both states where $p,q,r$ are all true satisfy that $\dia{?}(\Box_ap \et \neg\Box_bp)$. In the top-left $pqr$-state this is true because $\dia{q}(\Box_ap \et \neg\Box_bp)$ is true, whereas in the bottom-right $pqr$-state this is true because $\dia{r}(\Box_ap \et \neg\Box_bp)$ is true. However, there is no announcement $\phi$ such that $\dia{\phi}(\Box_ap \et \neg\Box_bp)$ is truth in both $pqr$-states. Assuming that there were such an announcement easily leads to a contradiction.
\center
\begin{tikzpicture}
\node (0r) at (6,0) {$\overline{p}q\overline{r}$};
\node (1r) at (8,0) {$pqr$};
\node (0ra) at (6,1.2) {$pqr$};
\node (1ra) at (8,1.2) {$\overline{pq}r$};
\draw[<->] (0r) -- node[above] {$a$} (1r);
\draw[<->] (0ra) -- node[above] {$a$} (1ra);
\draw[<->] (0r) -- node[left] {$b$} (0ra);
\draw[<->] (1r) -- node[right] {$b$} (1ra);
\draw[->] (0r) edge[loop left,looseness=9] node[left] {$ab$} (0r); 
\draw[->] (1r) edge[loop right,looseness=9] node[right] {$ab$} (1r); 
\draw[->] (0ra) edge[loop left,looseness=9] node[left] {$ab$} (0ra); 
\draw[->] (1ra) edge[loop right,looseness=9] node[right] {$ab$} (1ra); 
\end{tikzpicture}
}

For AUL this strong kind of synthesis is also not possible. But, somewhat surprisingly, in \cite{hales2013arbitrary}, Hales showed that this synthesis is possible for AML. This result was surprising for the following reason. Hales obtained his synthesis result with refinement modalities as quantifiers. It was already known that finite action model execution results in a refinement of the current relational model, but also that the other direction does not hold: there are refinements (i.e., updates) that can only be achieved by executing an infinite action model \cite{hvdetal.world:2008}. However, as the synthesis is with respect to making a given formula $\phi$ true, a finite syntactic object, synthesis for AML was after all possible.

In this contribution we show that synthesis is also possible for AUML. That is, for a given goal formula $\phi$, we can construct an arrow update model $X$ such that\begin{quote} For all $(M,s)$: there is an arrow update model $Y$ such that $(M,s)$ satisfies $\langle Y\rangle \phi$, if and only if $(M,s)$ satisfies $\langle X \rangle \phi$. \end{quote}
In AAUML we also have a quantifier over arrow update models. Therefore, in that logic the synthesis translates to the above-mentioned validity $ \dia{?}\phi \imp \dia{X}\phi$.  In AUML / AAUML we synthesize a (single-)pointed arrow update model, whereas for AAUML Hales synthesizes a {\em multi}-pointed action model, and it can be easily shown that this cannot be single-pointed.

\paragraph*{Results in the paper}
In this contribution we present {\em arbitrary arrow update model logic} (AAUML), that further extends arrow update model logic AUML, namely with dynamic modalities formalizing that {\em there is an arrow update {\bf model} after which} $\phi$. For this logic AAUML we obtain various results. We provide an {\em axiomatization} of AAUML. The axiomatization is a rewrite system allowing to eliminate dynamic modalities  from any given formula, while preserving truth. Thus, unlike AAUL, AAUML is decidable, and equally expressive as multi-agent modal logic. We establish {\em arrow update model synthesis}: if there is an arrow update model after which $\phi$, we can {\em construct} (synthesize) that model from $\phi$. We define a notion of {\em update expressivity} and we determine the relative update expressivity of AAUML and other arrow update logics and action model logics, and RML.

\paragraph*{Overview of content} Section \ref{sec.aauml} presents the syntax and semantics of arbitrary arrow update model logic, AAUML, and elementary structural notions. In Section~\ref{sec.synthesis} we describe the procedure for synthesizing arrow update models. In that section we also introduce a number of validities that are useful when introducing an axiomatization for AAUML, which we do in the subsequent  Section~\ref{sec.axiomatization}. Section \ref{sec.updateex} introduces the notion of {\em update expressivity}. Section \ref{sec.arrowvaction} compares AAUML and AAML, and in particular their update expressivity. This comparison also includes examples of arrow update models that have exponentially larger corresponding action models. Section \ref{sec.arrowvrefine} compares AAUML to RML.

\section{Arbitrary arrow update model logic} \label{sec.aauml}

Throughout this contribution, let $\Agents$ be a finite set of {\em agents} and let $\Atoms$ be a disjoint countably infinite set of {\em propositional variables} (or {\em atoms}).

\subsection{Structures} \label{sec.structures}

A {\em relational model} is a triple $M = (\States, R, V)$ with $S$ a non-empty {\em domain} (set) of {\em states} (also denoted $\Domain(M)$), $R$ a function assigning to each agent $a \in \Agents$ an {\em accessibility relation} $R_\agent \subseteq \States \times \States$, and $V: \Atoms \imp \States$ a {\em valuation function} assigning to each propositional variable
$p \in \Atoms$ the subset $V(p) \subseteq S$ where the variable is true. For $s \in \States$, the pair $(M,s)$ is called a {\em pointed relational model}, and for $T \subseteq \States$, the pair $(M,T)$ is called a {\em multi-pointed relational model}.

For any relation $R$ on a domain $X$, instead of $(x,y) \in R$ (where $x,y \in X$) we may write $R(x,y)$ or $xRy$, and $R(x)$ or $Rx$ for the set $\{ y \in X \mid R(x,y) \}$. If $R(x,y)$ we also say that $R$ \emph{links} $x$ to $y$, or that there \emph{is an arrow} from $x$ to $y$. Relation $R$ is: \emph{reflexive} iff for all $x \in X$, $R(x,x)$; \emph{serial} iff for all $x \in X$ there is $y \in X$ such that $R(x,y)$; \emph{transitive} iff for all $x,y,z \in X$, if $R(x,y)$ and $R(y,z)$ then $R(x,z)$; \emph{Euclidean} iff for all $x,y,z \in X$, if $R(x,y)$ and $R(x,z)$ then $R(y,z)$; an \emph{equivalence relation} iff it is reflexive, transitive, and Euclidean. Finally, for any $Y,Z \subseteq X$ we let $R(Y,Z)$ mean that for all $y \in Y$ there is a $z \in Z$ such that $R(y,z)$ and for all $z \in Z$ there is a $y \in Y$ such that $R(y,z)$; this is known as \emph{relational lifting}.

The class of relational models is known as $\mathcal{K}$. The class of relational models where all accessibility relations are equivalence relations is known as $\mathcal{S}5$, and the class of relational models where all accessibility relations are serial, transitive, and Euclidean is known as $\mathcal{KD}45$.

Let two relational models $M= (\States, R, V)$ and $M'= (\States', R', V')$ be given. A non-empty relation $\bisrel \subseteq \States \times \States'$ is a {\em bisimulation} if for all $(\state,\state') \in
\bisrel$ and $\agent\in\Agents$:
\begin{description}
\item[atoms] $\state \in V(p)$ iff $\state' \in V'(p)$ for all $p \in \Atoms$;
\item[forth] for all $\stateb \in \States$, if
$R_\agent(\state,\stateb)$, then there is a $\stateb'\in
\States'$ such that $R'_\agent(\state',\stateb')$ and
$(\stateb,\stateb') \in \bisrel$;
\item[back] for all $\stateb' \in \States'$,
if $R'_\agent(\state',\stateb')$, then there is a
$\stateb \in \States$ such that $R_\agent(\state,\stateb)$
and $(\stateb,\stateb') \in \bisrel$.
\end{description}
We write $M \bisim M'$ ({\em $M$ and $M'$ are bisimilar}) iff there is a bisimulation between $M$ and $M'$, and we write $(M,\state) \bisim (M',{\state'})$ ($(M,\state)$ and $(M',{\state'})$ are bisimilar) iff there is a bisimulation between $M$ and $M'$ linking $\state$ and $\state'$. Similarly, we write $(M,T) \bisim (M',T')$ iff there is a bisimulation between $M$ and $M'$ linking every state in $T$ to a state in $T'$ and linking every state in $T'$ to a state in $T$. 

Using the above-defined notion of relational lifting, if $\mathcal{M}_1$ and $\mathcal{M}_2$ are sets of pointed models we say that $\mathcal{M}_1$ and $\mathcal{M}_2$ are bisimilar, denoted $\mathcal{M}_1\bisim \mathcal{M}_2$, if for every $(M_1,s_1)\in \mathcal{M}_1$ there is an $(M_2,s_2)\in \mathcal{M}_2$ such that $(M_1,s_1)\bisim (M_2,s_2)$ and for every $(M_2,s_2)\in \mathcal{M}_2$ there is an $(M_1,s_1)\in \mathcal{M}_1$ such that $(M_1,s_2)\bisim (M_2,s_2)$.\footnote{For the purpose of bisimilarity, we could have treated a multi-pointed model $(M,T)$ as a set of pointed models $\{(M,t)\mid t\in T\}$, so that $(M_1,T_1)\bisim (M_2,T_2)$ if and only if for every $t_1\in T_1$ there is a $t_2\in T_2$ such that $(M_1,t_1)\bisim (M_2,t_2)$, and vice versa. As a union of bisimulations is again a bisimulation, that would have defined the same notion as above.}



We will now define arrow update models. We can think of them as follows. If you remove the valuation from a relational model you get a {\em relational frame}. We now decorate each arrow (pair in the accessibility relation for an agent) with two formulas in some logical language $\lang$: one for a condition that should hold in the source (state) of the arrow and the other that should hold in the target (state) of the arrow. The result is called an {\em arrow update model}.

\begin{definition}[Arrow update model]
Given a logical language $\lang$, an {\em arrow update model} $\aumodel$ is a pair $(\austates,\aufunction)$ where $\austates$ is a non-empty domain (set) of {\em outcomes} (also denoted $\Domain(\aumodel)$) and where $\aufunction$ is an {\em arrow relation} $\aufunction: \Agents \imp \powerset((\austates\times\lang) \times (\austates \times\lang))$. \end{definition} For each agent $\agenta$, the arrow relation links (outcome, formula) pairs to each other.  We write $\aufunction_\agent$ for $\aufunction(\agent)$, and we write $(o,\phi)\imp_\agent(o',\phi')$ for $((\austate,\phi),(\austate',\phi'))\in \aufunction_\agent$, or even, if the outcomes are unambiguous, $\phi\imp_\agent\phi'$. Formula $\phi$ is the {\em source condition} and formula $\phi'$ is the {\em target condition} of the $\agenta$-labelled {\em arrow} from {\em source} $o$ to {\em target} $o'$. A pointed arrow update model, or {\em arrow update}, is a pair $(U,o)$ where $o \in O$. Similarly, we define the {\em multi-pointed arrow update model $(U,Q)$}, where $Q \subseteq O$, known as well as arrow update. There is no confusion with the arrow updates of AUL \cite{kooietal:2011}, as those correspond to singleton pointed arrow update models.

Arrow update models are rather similar to the {\em action models} by Baltag {\em et al.}~\cite{baltagetal:1998}. They are compared in Section \ref{sec.arrowvaction}. 

\subsection{Syntax} \label{sec.syntax}

We proceed with the language and semantics of {\em arbitrary arrow update model logic} (AAUML). \begin{definition}[Syntax] \label{def.language} The language $\lang$ of AAUML consisting of \emph{formulas} $\phi$ is inductively defined as
\[
\lang \ \ni \ \  \phi ::= p \mid \neg \phi \mid (\phi \et \phi) \mid \Box_a \phi \mid [U,o]\phi \mid [\AAUL]\phi 
\]
where $p \in P$, $a \in A$, and where $U = (O, \aufunction)$ with $o \in \austates$ is an arrow update model with $O$ finite and with $\aufunction_a$ finite for all $a \in A$, and with source and target conditions that are formulas $\phi$.  \end{definition} The inductive nature of the definition may be unclear from the construct $[U,o]\phi$. We should think of $[U,o]\phi$ as an $n$-ary operator where not only the formula bound by $[U,o]$ is a formula but also all the source and target conditions in $U$.\footnote{The BNF informatics-style presentation obscures the inductive nature of the language definition, because the source and target conditions of $(U,o)$ are implicit. The mathematics-style presentation of that clause may be clearer: 

Let $\phi \in \lang$, let $U = (O,\aufunction)$ be an arrow update with source and target conditions $\phi_1,\dots,\phi_n \in\lang$ and such that $O$ and $\aufunction_a$ for all $a \in A$ are finite, and let $o \in O$. Then $[U,o]\phi\in\lang$.}
We read $[U,o]\phi$ as `after executing arrow update $(U,o)$, $\phi$ (holds), and $[\AAUL]\phi$ as `after an arbitrary arrow update, $\phi$ (holds)'. Other propositional connectives and dual diamond versions of modalities can be defined as usual by abbreviation: $\Dia_a \phi := \neg\Box_a \neg\phi$, $\dia{U,o}\phi := \neg [U,o] \neg \phi$, and $\dia{\AAUL}\phi := \neg [\AAUL]\neg\phi$. Expression $\phi[\psi/p]$ stands for uniform substitution of all occurrences of $p$ in $\phi$ for $\psi$. 

A formula is a \emph{modal formula} if it has shape $\Box_a\phi$, $\Dia_a\phi$, $[\AAUL]\phi$, $\dia{\AAUL}\phi$, $[U,o]\phi$, or  $\dia{U,o}\phi$. The \emph{modal depth} of a formula $\phi \in \lang$ is defined as: $d(p)=0$, $d(\neg\phi) = d(\phi)$, $d(\phi\et\psi) = \max(d(\phi),d(\psi))$, $d(\Box_a\phi) = d(\dia{\AAUL}\phi) = d(\phi)+1$, and $d([U,o]\phi = d(U)+d(\phi) +1$, where $d(U)$ is the maximum modal depth of the source and target conditions occurring in $U$. 

The propositional sublanguage is called $\langpl$ (the \emph{propositional formulas}). Adding the \emph{basic modal} construct $\Box_a\phi$ to $\langpl$ yields $\langml$ (the language of basic modal logic, the {\em basic modal formulas}). Additionally adding the construct $[U,o]\phi$ yields $\langauml$ (the language of arrow update model logic). 
In the language $\lang$ of AAUML, (modalities for) multi-pointed arrow update models are defined by abbreviation as $[U,Q]\phi := \Et_{o \in Q} [U,o]\phi$. From here on we also consider such modalities as logical connectives, such that $[U,Q]\phi$ is a formula in the logical language.


When doing synthesis, we will put formulas in {\em disjunctive negation normal form}. This fragment DNNF of $\lang$, that is inspired by the disjunctive normal form of propositional logic and the negation normal form of modal logic, is defined as
\[ \begin{array}{llll}
\mathrm{DNNF} & \ni &  \chi ::= \psi \mid (\chi\vel\chi) \\
&& \psi ::= \phi \mid (\psi \et \psi) \\
&& \phi ::= p \mid \neg p \mid \Box_a \chi \mid \Dia_a \chi \mid [U,o]\chi \mid \dia{U,o}\chi \mid [\AAUL]\chi \mid \dia{\AAUL}\chi
\end{array} \]
where the source and target conditions in $(U,o)$ are also formulas $\chi$.

This means that a $\phi \in \lang$ is in disjunctive negation normal form if \emph{every subformula} of $\phi$ is a disjunction of conjunctions of formulas that are an atom, or the negation of an atom, or that have one of $\square_a, \lozenge_a, [U,o], \langle U,o\rangle, [\AAUL]$ or $\langle\AAUL\rangle$ as main connective. In particular, this means that formulas have to be in DNNF at every modal depth. So, for example, $p\vee \square (q\vee (\lozenge p \wedge \neg q))$ is in DNNF, while $p\vee \square (q\vee \neg (\neg\lozenge p \vee q))$ is not. 

\subsection{Semantics} \label{sec.semantics}

We continue with the semantics.  The semantics are defined by induction on $\phi \in \lang$, and simultaneously with the execution of arrow update models.
\begin{definition}[Semantics] \label{def.semantics}
Let a relational model $M = (\States,R,V)$, a state $s \in \States$, an arrow update model $U = (O,\aufunction)$, and a formula $\phi\in\lang$ be given. The truth (or satisfaction) of $\phi$ in $(M,s)$ is defined by induction on $\phi$.
\[ \begin{array}{lcl} 
M,s \models \atom & \text{ iff } & s \in V(\atom) \\ 
M,s \models \neg \phi & \text{ iff } & M,s \not\models \phi  \\ 
M,s \models \phi\et\psi & \text{ iff } & M,s \models \phi \text{ and } M,s \models \psi  \\ 
M,s \models \Box_a \phi & \text{ iff } & M,t \models \phi \text{ for all } (s,t) \in R_\agent \\
M,\state \models [\aumodel,\austate] \phi & \text{ iff } & M * \aumodel, (\state,\austate) \models \phi \text{ where $M * \aumodel$ is defined in $(\sharp)$}  \\
M,s \models [\AAUL] \phi & \text{ iff } & M,s \models [U,o] \phi \text{ for all $(U,o)$ satisfying $(\sharp\sharp)$} \\
\end{array} \]
$(\sharp)$: $M * \aumodel = (\States',R',V')$ is defined as
\[ \begin{array}{lll}
\States' & = & \States \times \austates \\ 
&& \text{For all } a \in A, \phi,\phi'\in\lang, s,s'\in S, o,o'\in O: \\
((\state,\austate), (\state',\austate')) \in R'_\agent & \text{iff} & (\state,\state') \in R_\agent, (\austate,\phi)\imp_\agent(\austate',\phi'), M,\state \models \phi, \text{ and } M,\state' \models \phi' \\
&& \text{For all }p \in P: \\
V'(\atom) & = & V(\atom) \times \austates
\end{array} \]
$(\sharp\sharp)$: $(U,o)$ is an arrow update with all source and target conditions in $\langml$. 

Formula $\phi$ is {\em valid in $M$}, notation $M \models \phi$, iff $M,s \models \phi$ for all $s \in S$; and $\phi$ is {\em valid} iff for all relational models $M$ we have that $M \models \phi$. Formulas $\phi,\psi\in\lang$ are \emph{equivalent} iff for all $M = (S,R,V)$ and for all $s \in S$, $M,s \models \phi$ iff $M,s \models\psi$. The set of validities, also more properly known as the \emph{logic}, is called AAUML.
\end{definition}

Formulas $\phi$ and $\psi$ from different languages will also be called equivalent if they satisfy the above condition. The term AAUML will also continue to be informally used for arbitrary arrow update logic.
The restriction of arrow formulas to $\langml$ in the semantics of $[\AAUL] \phi$ is to avoid circularity of the semantics, as $[\AAUL] \phi$ could otherwise itself be one of those arrow formulas. However, because we will prove that AAUML is as expressive as basic modal logic, we also have
\begin{equation*}M,s\models [\AAUL]\phi \text{ \quad iff \quad } M,s\models [U,o]\phi \text{ for all } (U,o) \hspace{3cm} \ \end{equation*}
without any restriction on the source and target conditions of $U$. We will prove this property in Proposition \ref{prop:fully_arbitrary}, later. 

We conclude this subsection by noting two relatively simple properties of AAUML that will be useful in later sections. Firstly, AAUML is invariant under bisimulation, i.e., if $(M,s)\bisim(M',s')$ then for all $\phi\in \lang$ we have $M,s\models \phi$ iff $M',s'\models \phi$. In \cite[Lemma~3]{hvdetal.aij:2017} it was shown that AUML is invariant under bisimulation. The proof given in \cite{hvdetal.aij:2017} can easily be extended to a proof that AAUML is also invariant under bisimulation.

Secondly, every formula $\phi \in\lang$ is equivalent to a formula $\phi'$ that is in DNNF. Proving the existence of such $\phi'$ is conceptually simple but rather notationally complex. We therefore provide only an example, and trust that the reader can see that the demonstrated process can be generalized to any $\phi \in \lang$. Suppose that $\phi = p \wedge \neg (\square_a\psi_1 \wedge \neg [U,o]\psi_2)$. Our first step is to treat the non-propositional subformulas of $\phi$ as atoms, i.e., we treat the formula as $p\wedge \neg (q_1\wedge \neg q_2)$. This is a formula of propositional logic, so it is equivalent to a formula in disjunctive normal form: $(p\wedge \neg q_1)\vee (p\wedge q_2)$. Then we recall what $q_1$ and $q_2$ represent, and obtain $(p\wedge \neg \square_a\psi_1)\vee (p\wedge [U,o]\psi_2)$. Using the fact that $\neg \square_a\psi_1$ is equivalent to $\lozenge_a\neg \psi_1$, we find that $\phi$ is equivalent to $(p\wedge \lozenge_a\neg \psi_1)\vee (p\wedge [U,o]\psi_2)$. We then repeat this process for $\neg \psi_1$, $\psi_2$ and every formula $\chi$ that occurs as a source or target condition in $U$. The depth of $\neg \psi_1, \psi_2$ and every such $\chi$ is strictly lower than that of $\phi$, so this process eventually terminates, resulting in a formulas $\psi_1', \psi_2'$ and $\chi'$ that are in DNNF and equivalent to $\neg \psi_1, \psi_2$ and $\chi$, respectively. Let $U'$ be the result of replacing every $\chi$ in $U$ by the equivalent $\chi'$. Then the formula $(p\wedge \lozenge_a\psi_1')\vee (p\wedge [U',o]\psi_2')$ is in DNNF and equivalent to $\phi$.


\subsection{Example} \label{sec.aaumlex}

First consider the action of the introductory section of Anne reading a letter containing the truth about $p$ while Bill remains uncertain whether she performs that action. The arrow update model producing the resulting information state is depicted in the upper part of Figure~\ref{fig.below}. In the figure, an arrow $\imp$ labelled with $\phi \ _i \ \phi'$ and linking outcomes $o,o'$ stands for the arrow $\phi \imp_i \phi'$ between these outcomes, i.e., $((o,\phi), (o',\phi')) \in \aufunction_i$; $\phi \ _{ij} \ \phi'$ stands for $\phi \imp_i \phi'$ and $\phi \imp_j \phi'$.

In the resulting model Bill considers it possible that Anne knows $p$, that she knows $\neg p$, and that she still is uncertain about $p$: $\Dia_b \Box_a p \et \Dia_b\Box_a \neg p \et \Dia_b\neg(\Box_a p \vel \Box_a \neg p)$.

Next, consider the action of Anne privately learning that $p$ while Bill remains unaware of her doing so. The arrow update model achieving that and the resulting relational model are depicted in the lower part of Figure~\ref{fig.below}. In the resulting model it is true that, for example, Anne believes $p$ but Bill incorrectly believes that Anne is uncertain about $p$: $\Box_a p \et \Box_b \neg(\Box_a p \vel \Box_a \neg p)$. 

\begin{figure}[h]
\center
\begin{tikzpicture}
\node (0) at (0,0) {$\neg p$};
\node (1) at (2,0) {$\fbox{$p$}$};
\draw[<->] (0) -- node[above] {$ab$} (1);
\draw[->] (0) edge[loop above,looseness=15] node[above] {$ab$} (0); 
\draw[->] (1) edge[loop above,looseness=9] node[above] {$ab$} (1); 
\node (t) at (3,0) {\large $*$};
\node (0r) at (5,0) {$\bullet$};
\node (1r) at (5,2) {$\fbox{$\circ$}$};
%
\draw[<->] (0r) -- node[right] {$\T \ _b \ \T$} (1r);
\draw[->] (0r) edge[loop right,looseness=12] node[right] {$\T \ _{ab} \ \T$} (0r); 
\draw[->] (1r) edge[loop right,looseness=9] node[right] {$p \ _a \ p$} (1r); 
\draw[->] (1r) edge[loop above,looseness=9] node[above] {$\neg p \ _a \ \neg p$} (1r); 
\draw[->] (1r) edge[loop left,looseness=9] node[left] {$\T \ _b \ \T$} (1r); 
%
%
\node (t) at (8,0) {\large $=$};
\node (0r) at (11,0) {$(\neg p, \bullet)$};
\node (1r) at (13.5,0) {$(p, \bullet)$};
\node (0ra) at (11,2) {$(\neg p, \circ)$};
\node (1ra) at (13.5,2) {$\fbox{$(p, \circ)$}$};
\draw[<->] (0r) -- node[above] {$ab$} (1r);
\draw[<->] (0ra) -- node[above] {$b$} (1ra);
\draw[<->] (0r) -- node[left] {$b$} (0ra);
\draw[<->] (1r) -- node[right] {$b$} (1ra);
\draw[->] (0r) edge[loop left,looseness=4] node[left] {$ab$} (0r); 
\draw[->] (1r) edge[loop right,looseness=4] node[right] {$ab$} (1r); 
\draw[->] (0ra) edge[loop left,looseness=4] node[left] {$ab$} (0ra); 
\draw[->] (1ra) edge[loop right,looseness=4] node[right] {$ab$} (1ra); 
\end{tikzpicture}

\bigskip

\noindent
\begin{tikzpicture}
\node (0) at (0,0) {$\neg p$};
\node (1) at (2,0) {$\fbox{$p$}$};
\draw[<->] (0) -- node[above] {$ab$} (1);
\draw[->] (0) edge[loop above,looseness=15] node[above] {$ab$} (0); 
\draw[->] (1) edge[loop above,looseness=9] node[above] {$ab$} (1); 
\node (t) at (3,0) {\large $*$};
\node (0r) at (5,0) {$\bullet$};
\node (1r) at (5,2) {$\fbox{$\circ$}$};
%
\draw[<-] (0r) -- node[right] {$\T \ _b \ \T$} (1r);
\draw[->] (0r) edge[loop right,looseness=12] node[right] {$\T \ _{ab} \ \T$} (0r); 
\draw[->] (1r) edge[loop right,looseness=9] node[right] {$\T \ _a \ p$} (1r); 
%
%
\node (t) at (8,0) {\large $=$};
\node (0r) at (11,0) {$(\neg p, \bullet)$};
\node (1r) at (13.5,0) {$(p, \bullet)$};
\node (0ra) at (11,2) {$(\neg p, \circ)$};
\node (1ra) at (13.5,2) {$\fbox{$(p, \circ)$}$};
\draw[<->] (0r) -- node[above] {$ab$} (1r);
\draw[->] (0ra) -- node[above] {$a$} (1ra);
\draw[<-] (0r) -- node[left] {$b$} (0ra);
\draw[<-] (1r) -- node[right] {$b$} (1ra);
\draw[->] (0r) edge[loop left,looseness=4] node[left] {$ab$} (0r); 
\draw[->] (1r) edge[loop right,looseness=4] node[right] {$ab$} (1r); 
\draw[->] (1ra) edge[loop right,looseness=4] node[right] {$a$} (1ra); 
\end{tikzpicture}
\caption{Different ways of Anne learning that $p$}
\label{fig.below}
\end{figure}
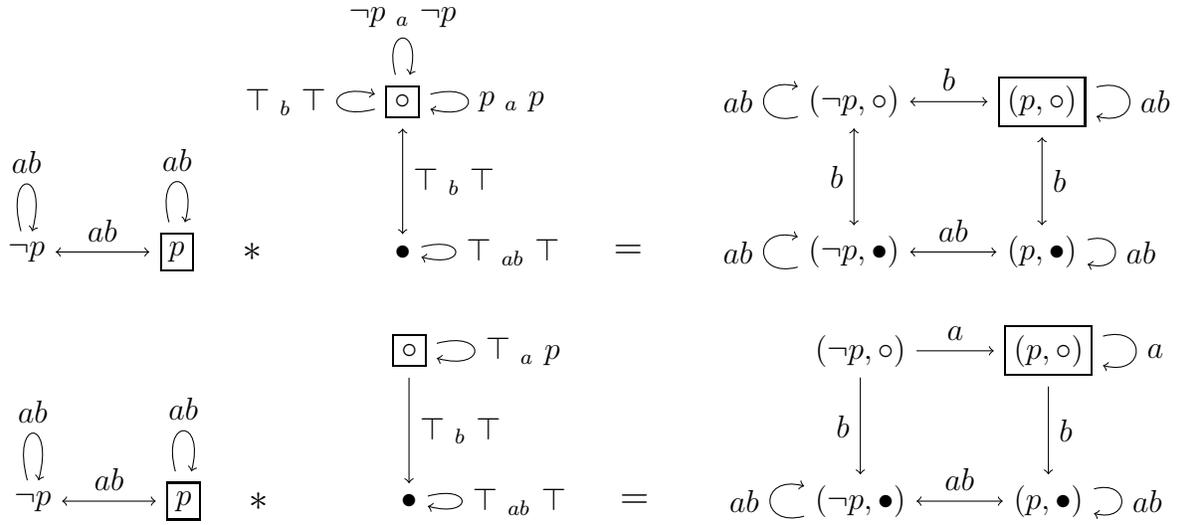

The relation $\aufunction_a$ allows for multiple pairs between the same outcomes. This is necessary. For an example, the singleton arrow update with two reflexive arrows $p \imp_a q, r \imp_a s$ (i.e., $(o,p)\aufunction_a(o,q)$ and $(o,r)\aufunction_a(o,s)$) does not correspond to an arrow update where for any given agent $a$ at most a single arrow links any two outcomes, see \cite{kooietal:2011,kooirenne,hvdetal.aij:2017}.

Arrow updates apply to any kind of relational model, and also in particular to relational models wherein all accessibility relations are equivalence relations, the class $\mathcal{S}5$. These relations model {\em knowledge} of an agent. Relational models wherein all accessibility relations are serial, transitive, and Euclidean, are of the class $\mathcal{KD}45$. These relations model {\em consistent belief} of an agent. As dynamic epistemic logics typically formalize change of knowledge or change of belief, i.e., {\em epistemic change}, of particular interest are therefore arrow updates that are $\mathcal{S}5$-preserving or $\mathcal{KD}45$-preserving, by which we mean that, given a relational model in class $\mathcal{S}5$, the update will produce a relational model in class $\mathcal{S}5$, and similarly for $\mathcal{KD}45$.

The examples in this section are indeed typical in that sense. The first example is an $\mathcal{S}5$-preserving update and the second example is a  $\mathcal{KD}45$-preserving update.

There is more to be learnt from these examples: the first arrow update `seems $\mathcal{S}5$' and the second update `seems $\mathcal{KD}45$'. It is easy to make `seem' precise: consider the following accessibility relation between outcomes induced by an arrow relation: \[ o \imp'_a o' \quad \text{ iff } \quad (o,\phi) \imp_a (o',\phi') \text{ for some } \phi,\phi'. \] Let an arrow update be {\em in class $\mathcal{S}5$} if for all agents $a$, these induced $\imp'_a$ are equivalence relations; and similarly for $\mathcal{KD}45$. The first arrow update is therefore $\mathcal{S}5$ and the second arrow update is $\mathcal{KD}45$. However, an $\mathcal{S}5$ arrow update of an $\mathcal{S}5$ relational model may not result in an $\mathcal{S}5$ relational model (whereas an $\mathcal{S}5$ action model executed in an $\mathcal{S}5$ relational model will always result in an $\mathcal{S}5$ relational model). This is obvious, as the presence of arrows in the resulting model is also determined by source and target conditions. For example, if in the arrow update of the first example we change the arrow $\top \imp_b \top$ linking $\circ$ to $\bullet$ into $\bot \imp_b \top$, then the resulting model will no longer be reflexive. It is no longer $\mathcal{S}5$. It is not known how to address such issues systematically (see Section \ref{sec.conclusion}).

As said, arrow updates are an alternative modelling mechanism to the better known action models. In Section \ref{sec.arrowvaction}, and in particular Subsection \ref{sec.applications}, we compare the two mechanisms in more detail, we will give action models that define the same update as the arrow updates in this section, and we will also present typical applications on which they perform differently.

\section{Arrow update synthesis} \label{sec.synthesis}

The goal of synthesis for AUML is to find, given a goal formula $\phi$, an arrow update (i.e., a pointed arrow update model) $(U,o)$ that makes $\phi$ true. There are at least three ways in which we could interpret this goal, however.
\begin{definition}[Synthesis] \ 
\begin{itemize}
	\item The \emph{local synthesis problem} takes as input a pointed model $(M,s)$ and a goal formula $\phi$. The output is an arrow update $(U,o)$ such that $M,s\models \langle U,o\rangle \phi$, or ``NO'' if no such arrow update exists.
	\item The \emph{valid synthesis problem} takes as input a goal formula $\phi$. The output is an arrow update $(U,o)$ such that $\models \langle U,o\rangle \phi$, or ``NO'' if no such arrow update exists.
	\item The \emph{global synthesis problem} takes as input a goal formula $\phi$. The output is an arrow update $(U,o)$ such that for every pointed model $(M,s)$, if there is some $(U',o')$ such that $M,s\models \langle U',o'\rangle \phi$, then $M,s\models \langle U,o\rangle \phi$.
\end{itemize}
\end{definition}
We recall from the introduction that we take the third approach: when we say synthesis we mean global synthesis. An alternative, equivalent characterization of the global synthesis problem is that, for given $\phi$, we want to find $(U,o)$ such that, for all $(M,s)$, $M,s\models {\langle \AAUL\rangle \phi} \leftrightarrow \langle U,o\rangle \phi$ (see Proposition \ref{prop:fully_arbitrary}). 

Note that for the global synthesis problem, unlike the local synthesis problem and the valid synthesis problem, we do not allow ``NO'' as an output. As a result, it is not obvious that global synthesis for AUML is possible at all. We also recall from the introduction that synthesis is impossible for PAL and for AUL, but possible for AML \cite{hales2013arbitrary}. We now show that synthesis for AUML is indeed also possible. Because our version of synthesis is global, it cannot depend on any specific model. So our synthesis process is purely syntactic.

In our synthesis, we make use of so-called \emph{reduction axioms}. These reduction axioms are a set of validities that, when taken together, show that AAUML has the same expressive power as modal logic.

\subsection{Reduction axioms for arrow update models}
\label{sec:red_U}
We start by considering the reduction axioms for the $[U,o]$ operator.
\begin{lemma}[{\cite{kooirenne}}]
Let $(U,o)$ be an arrow update, $p\in P$, $a\in A$ and $\phi,\psi \in \lang$. Then the following validities hold.
\begin{align*}
\models {} & [U,o]p \leftrightarrow p\\
\models {} & [U,o]\neg \phi \leftrightarrow \neg [U,o]\phi\\
\models {} & [U,o](\phi\wedge \psi)\leftrightarrow ([U,o]\phi \wedge[U,o]\psi)\\
\models {} & [U,o]\square_a\phi \leftrightarrow \bigwedge_{(o,\psi)\rightarrow_a (o',\psi')}(\psi\rightarrow \square_a(\psi'\rightarrow [U,o']\phi))
\end{align*}
\end{lemma}
\begin{proof}

The first three validities follow immediately from the semantics of $[U,o]$. The fourth validity also follows from the semantics, in the following way.
\[ \begin{array}{lll}
M,w\models [U,o]\square_a\phi & \text{iff} \ \ & M*U,(w,o)\models \square_a\phi\\
&\text{iff} & \text{for all}\ (w',o') \ \text{such that}\ (w,o)R_a(w',o'): M*U, (w',o')\models \phi\\
&\text{iff} & \text{for all}\ (o',\psi') \ \text{and}\ w' \ \text{such that}\ (o,\psi)\rightarrow_a (o',\psi') \ \text{and}\ wR_aw': \\ && \text{if } M,w\models \psi \ \text{and}\ M,w'\models \psi' \text{ then } M*U,(w',o')\models \phi\\
&\text{iff} & \text{for all}\ (o',\psi') \ \text{such that}\ (o,\psi)\rightarrow_a (o',\psi'): \\ && \text{if } M,w\models \psi  \text{ then } M,w\models \square_a (\psi'\rightarrow[U,o']\phi)\\
&\text{iff} & M,w\models \bigwedge_{(o,\psi)\rightarrow_a (o',\psi')}(\psi\rightarrow \square_a(\psi'\rightarrow [U,o']\phi))  \vspace{-.7cm}
\end{array}\] 

\end{proof}
Note that, in particular, $\models [U,o]\neg \phi \leftrightarrow \neg [U,o]\phi$ implies that $[U,o]$ is self-dual: we have $\models [U,o]\phi \leftrightarrow \langle U,o\rangle \phi$. This, of course, does not extend to the arbitrary arrow update operator: there are $\phi$ for which $\not \models [\AAUL]\phi \leftrightarrow \langle \AAUL\rangle \phi$.

The above lemma shows that $[U,o]$ commutes with $\neg$, distributes over $\wedge$ and, in a somewhat complicated way, commutes with $\square_a$. As discussed in \cite{kooirenne}, this suffices to show that $[U,o]$ can be eliminated from the restriction of the language $\lang$ to $\langauml$.
\begin{corollary}
\label{cor:auml_reduction}
For every $\phi \in \langauml$ there is a formula $\phi'\in \langml$ such that $\models \phi \leftrightarrow \phi'$.
\end{corollary}

\subsection{Reduction axioms for the arrow update model quantifier}
\label{sec:red_AAUL}
We can also write similar reduction axioms for $[\AAUL]$. In practice, however, it turns out to be slightly more convenient to write them for the dual operator $\langle \AAUL\rangle$. Note that in the lemmas in this subsection we restrict ourselves to the language $\langauml$, as some of those lemmas use that $\langle \AAUL\rangle$ quantifies over arrow updates with source and target conditions in $\langml$, and because we can meet this constraint by applying Corollary~\ref{cor:auml_reduction}. Later, in Theorem~\ref{thm:expressivity} in the next subsection, we will show that this restriction is unnecessary, and that the lemmas apply to $\lang$ as well.

\begin{lemma} \label{lemma:reduction_simple_2a}
For every $\phi \in \langauml$ and every $a\in A$, we have
\[\models \langle \AAUL\rangle \lozenge_a\phi \leftrightarrow \lozenge_a\langle \AAUL\rangle\phi\]
\end{lemma}
\begin{proof}
Let $(\M,w)$ be any pointed model, and suppose that $\M,w\models \langle\AAUL\rangle \lozenge_a\phi$. Then there is some $(U,o)$ such that $\M*U,(w,o)\models \lozenge_a\phi$. So $(w,o)$ has an $a$-successor $(w',o')$ such that $\M*U,(w',o')\models \phi$.

This implies that $\M,w'\models \langle U,o'\rangle \phi$ and therefore $\M,w'\models \langle \AAUL\rangle\phi$. Since $w'$ is an $a$-successor of $w$, we obtain $\M,w\models \lozenge_a\langle\AAUL\rangle \phi$.

Now, suppose that $\M,w\models \lozenge_a\langle \AAUL\rangle\phi$. Then there is an $a$-successor $w'$ of $w$ such that $\M,w'\models \langle \AAUL\rangle \phi$. As witness for this $\langle \AAUL\rangle$ statement there must be some $U',o'$ such that $\M,w'\models \langle U',o'\rangle \phi$.

Let $(U,o)$ be the arrow update obtained by adding one extra world $o$ to $U'$, and a transitions $(o,\top)\rightarrow_a (o',\top)$. Note that $(\M*U',(w',o'))$ is bisimilar to $(\M*U,(w',o'))$, and therefore $\M*U,(w',o')\models \phi$. Finally, note that $(w',o')$ is an $a$-successor of $(w,o)$, so we have $\M*U,(w,o)\models \lozenge_a\phi$ and therefore $\M,w\models \langle \AAUL \rangle\lozenge_a\phi$.
\end{proof}
Note that the proof is constructive. That is, if we find $(U',o')$ such that $\M,w\models \lozenge_a\langle U',o'\rangle \phi$ then not only do we know that $\M,w\models \langle \AAUL\rangle \lozenge_a\phi$, we can also find a specific $(U,o)$ such that $\M,w\models \langle U,o\rangle \lozenge_a\phi$.

Next, we consider a slightly stronger lemma.
\begin{lemma}
\label{lemma:reduction_simple_2}
For every $\phi_1,\cdots,\phi_n\in \langauml$ and every $a\in A$ we have
\begin{equation*}
\models \langle \AAUL\rangle \bigwedge_{1\leq i \leq n}\lozenge_a\phi_i \leftrightarrow \bigwedge_{1\leq i \leq n}\lozenge_a\langle \AAUL\rangle\phi_i
\end{equation*}
\end{lemma}
\begin{proof}
The left-to-right direction is obvious, so we show only the right-to-left direction. So suppose that $\M,w\models \bigwedge \lozenge_a\langle \AAUL\rangle \phi_i$. Then there are $a$-successors $w_1,\cdots,w_n$ of $w$ and pointed arrow update models $(U_1,o_1),\cdots, (U_n,o_n)$ such that $\M,w_i\models \langle U_i,o_i\rangle\phi_i$ for all $i$.

Now, let $(U,o)$ be the arrow update obtained by taking the disjoint union of all $U_i$ and adding one extra outcome $o$, and adding arrows $(o,\top)\rightarrow_a (o_i,\top)$ for every  $o_i$.

For every $i$, $(\M*U_i,(w_i,o_i))$ is bisimilar to $(\M*U,(w_i,o_i))$, so we have $\M*U,(w_i,o_i)\models \phi_i$. Finally, $(w_i,o_i)$ is an $a$-successor of $(w,o)$ for every $i$. As such, we have $\M,w\models \langle U,o\rangle \bigwedge\lozenge_a\phi_i$ and therefore, as all the source and target conditions of $U$ are in $\langml$, $\M,w\models \langle \AAUL\rangle \bigwedge\lozenge_a\phi_i$.
\end{proof}
Again, the proof is constructive, so given $(U_i,o_i)$ for all $i$, we can find the model $(U,o)$. Note also that the $\phi_i$ need not be consistent with each other.

Some reflection may be in order as to why Lemma~\ref{lemma:reduction_simple_2} holds. Suppose that $\M,w\models \bigwedge \lozenge_a\langle \AAUL\rangle\phi_i$. So for every $i$, there is some world $w_i$ that $a$ considers possible as well as some event $U_i$ and outcome $o_i$ such that, if $(U_i,o_i)$ were to happen in $w_i$, then $\phi_i$ would become true.

Now let us look at the arrow update $(U,o)$ that we constructed. Effectively, this arrow update represents us telling $a$ that ``we are performing one of the actions $U_i,o_i$, but we are not telling you which one.'' Now, for every $i$ agent $a$ considers it possible that $w_i$ is the actual world, and that $(U_i,o_i)$ is the event that happened. As such, after we execute our event we are in a situation where every $\phi_i$ is held possible by $a$.

So far, we have only considered diamonds. Now, let us add a box modality.
\begin{lemma}
\label{lemma:AAUL_reduction_1}
For every $\phi_1,\cdots,\phi_n,\psi\in \langauml$ and every $a\in A$, we have
\begin{equation*}\models\langle \AAUL\rangle (\bigwedge_{1\leq i \leq n}\lozenge_a\phi_i\wedge \square_a\psi)\leftrightarrow \bigwedge_{1\leq i \leq n}\lozenge_a \langle \AAUL\rangle(\phi_i\wedge \psi)\end{equation*}
\end{lemma}
\begin{proof}
The left-to-right direction is fairly obvious. Suppose that $M,w \models\langle \AAUL\rangle (\bigwedge_{1\leq i \leq n}\lozenge_a\phi_i\wedge \square_a\psi)$. Then there is a $(U,o)$ such that $M,w \models \dia{U,o} (\bigwedge_{1\leq i \leq n}\lozenge_a\phi_i\wedge \square_a\psi)$. Therefore,  $M*U,(w,o) \models \Box_a \psi$ and for each $1\leq i \leq n$, $M*U,(w,o) \models \lozenge_a\phi_i$. Let $(w_i,o')$ be such that $(w,o)R_a(w_i,o')$ and  $M*U,(w_i,o') \models\phi_i$. From $M*U,(w,o) \models \Box_a \psi$ and $(w,o)R_a(w_i,o')$ also follows that $M*U,(w_i,o') \models\psi$. Combining both we have $M*U,(w_i,o') \models\phi_i \et \psi$. Therefore, $M,w_i \models\dia{U,o'} (\phi_i \et \psi)$, from which it follows that $M,w_i \models\langle \AAUL\rangle(\phi_i \et \psi)$. From $(w,o)R_a(w_i,o')$ it follows by definition that $wR_aw_i$. From $M,w_i \models\langle \AAUL\rangle(\phi_i \et \psi)$ and $wR_aw_i$ we get the required $M,w \models\Dia_a\langle \AAUL\rangle(\phi_i \et \psi)$. As $i$ was arbitrary, $M,w \models\Et_{1 \leq i \leq n} \Dia_a\langle \AAUL\rangle(\phi_i \et \psi)$.

We now show the right-to-left direction. So suppose that $\M,w\models \bigwedge_{1\leq i \leq n} \lozenge_a\langle \AAUL\rangle (\phi_i\wedge\psi)$. Then for every $1\leq i \leq n$, there are an $a$-successor $w_i$ of $w$ and $(U_i,o_i)$ such that $\M,w_i\models \langle U_i,o_i\rangle (\phi_i\wedge \psi)$.

Let $(U,o)$ be the model obtained by taking the disjoint union of all $U_i$, and adding a single outcome $o$ with arrows $(o,\top)\rightarrow_a (o_i,\langle U_i,o_i\rangle\psi)$ for every $i$.

Consider $(\M*U,(w,o))$. By assumption, $\M,w_i\models \langle U_i,o_i\rangle (\phi_i\wedge \psi)$, so $\M,w_i\models \langle U_i,o_i\rangle \psi$. From that and the fact that model $U$ contains arrow $(o,\top)\rightarrow_a (o_i,\langle U_i,o_i\rangle\psi)$ it follows that $(w_i,o_i)$ is an $a$-successor of $(w,o)$ in $(\M*U)$. Furthermore, also from $\M,w_i\models \langle U_i,o_i\rangle (\phi_i\wedge \psi)$ it follows that $\M,w_i\models \langle U_i,o_i\rangle\phi_i$, i.e.,  $\M*U_i,(w_i,o_i)\models \phi_i$, and as $(\M*U,(w_i,o_i))$ is bisimilar to $(\M*U_i,(w_i,o_i))$ it follows that $\M*U,(w_i,o_i)\models \phi_i$ (the bisimulation is the identity relation on the restriction of the domain of $\M*U$ to the domain of $\M*U_i$). With $(w,o)R_a(w_i,o_i)$ we thus get $\M*U,(w,o)\models \Dia_a\phi_i$, and as $i$ was arbitrary, $\M*U,(w,o)\models \Et_{1 \leq i \leq n}\Dia_a\phi_i$.

Additionally, note that every outgoing $a$-arrow in $(U,o)$ has target condition $\langle U_i,o_i\rangle\psi$ for some $i$. Again observing that $(\M*U,(w_i,o_i))$ is bisimilar to $(\M*U_i,(w_i,o_i))$, it follows that for every $(w_i,o_i)$ that is an $a$-successor of $(w,o)$, we have $\M*U,(w_i,o_i)\models \psi$. It follows that $\M*U,(w,o)\models \square_a\psi$.

Taken together, the above shows that $\M*U,(w,o) \models \bigwedge_{1\leq i \leq n} \lozenge_a\phi_i\wedge\square_a\psi$ and thus that $\M,w\models \langle U,o\rangle (\bigwedge_{1\leq i \leq n} \lozenge_a\phi_i\wedge\square_a\psi)$. Furthermore, every formula $\chi$ occurring in $U$ is either a  basic modal formula, or of the form $\langle U_i,o_i\rangle \psi$. Because $\psi\in \langauml$, we have $\langle U_i,o_i\rangle \psi\in\langauml$ and therefore, by Corollary~\ref{cor:auml_reduction}, there is a formula $\chi_i\in \langml$ that is equivalent to $\langle U_i,o_i\rangle \psi$. Let $U'$ be the arrow update model obtained from $U$ by replacing $\langle U_i,o_i\rangle \psi$ by $\chi_i$, for every $i$.
Then we also have $\M,w\models \langle U',o\rangle (\bigwedge_{1\leq i \leq n} \lozenge_a\phi_i\wedge\square_a\psi)$. Because $U'$ only contains formulas of modal logic, it follows that $\M,w\models \langle \AAUL\rangle (\bigwedge_{1\leq i \leq n} \lozenge_a\phi_i\wedge\square_a\psi)$, as was to be shown.
\end{proof}
Once again, the proof is constructive. Note that on the right-hand side we have eliminated the $\square_a$ connective. This is a consequence of the fact that as the designer of the arrow update model $U$, we have the freedom to inform $a$ that certain worlds, which she might previously have considered possible, are not in fact the actual world. This results in the removal of the $a$-arrows to these worlds. So if we want to make $\square_a\psi$ true after the application of $U$, then we can simply have $a$ eliminate all successors where $\psi$ would otherwise become false. In the construction used in the lemma, we do this using the target condition $\langle U_i,o_i\rangle \psi$.

In the preceding three lemmas, we only considered $\lozenge_a$ and $\square_a$ operators for one single agent $a$. However, when constructing $(U,o)$ we can place arrows for different agents independently, so the same construction works for multiple agents at the same time. This yields the following lemma.

\begin{lemma}
\label{lemma:AAUL_main_reduction}
For every $a\in A$, let $\Phi_a\subseteq \langauml$ be a finite set of  formulas, and let $\psi_a\in \langauml$ be a formula. Then
\begin{align*}\models \langle \AAUL\rangle\bigwedge_{a\in A}(\bigwedge_{\phi_a \in \Phi_a}\lozenge_a\phi_a\wedge \square_a\psi_a)\leftrightarrow & \bigwedge_{a\in A}\bigwedge_{\phi_a\in \Phi_a} \lozenge_a\langle\AAUL\rangle(\phi_a\wedge\psi_a)
\end{align*}
\end{lemma}

\begin{proof}
We proceed as in the previous Lemma \ref{lemma:AAUL_reduction_1}, but simultaneously for all $a \in A$. 

Again, the left-to-right direction is fairly obvious, and completely analogous to the treatment in Lemma \ref{lemma:AAUL_reduction_1}. Suppose that $M,w \models\langle \AAUL\rangle\bigwedge_{a\in A}(\bigwedge_{\phi_a \in \Phi_a}\lozenge_a\phi_a\wedge \square_a\psi_a)$. Then there is a $(U,o)$ such that $M,w \models \dia{U,o} (\bigwedge_{a\in A}(\bigwedge_{\phi_a \in \Phi_a}\lozenge_a\phi_a\wedge \square_a\psi_a)$. Therefore, for all $a \in A$ and for all $\phi_a \in \Phi_a$,  $M*U,(w,o) \models \Box_a \psi$ and $M*U,(w,o) \models \lozenge_a\phi_a$. As before we now obtain $M,w \models \Dia_a \dia{\AAUL} (\phi_a \et \psi)$,  for all $a \in A$ and for all $\phi_a \in \Phi_a$, and thus $M,w \models \Et_{a \in A} \Et_{\phi_a \in \Phi_a} \Dia_a \dia{\AAUL} (\phi_a \et \psi)$.

From right to left, suppose that $\M,w\models \bigwedge_{a\in A}\bigwedge_{\phi_a\in \Phi_a} \lozenge_a\langle \AAUL\rangle (\phi_a\wedge\psi_a)$. Let $\Phi_a = \{ \phi_a^1,\dots,\phi_a^{|\Phi_a|} \}$ ($\Phi_a$ may be empty). Then for every $a \in A$ and for every $1\leq i \leq |\Phi_a|$, there are $(U_a^i,o_a^i)$ and an $a$-successor $w_a^i$ of $w$ such that $\M,w_a^i\models \langle U_a^i,o_a^i\rangle (\phi_a^i\wedge \psi_a)$. Similarly to the previous lemma, we let $(U,o)$ be the model obtained by taking the disjoint union of all $U_a^i$ (i.e., for all $a \in A$ and for all $1\leq i \leq |\Phi_a|$), but still only adding a \emph{single} outcome $o$ with arrows $(o,\top)\rightarrow_a (o_i,\langle U_i,o_i\rangle\psi)$ for every $a \in A$ and $1\leq i \leq |\Phi_a|$. We then proceed as before.
\end{proof}
Lemma~\ref{lemma:AAUL_main_reduction} is the most important reduction axiom for AAUML. However, not every formula is of a form such that the lemma can be applied. We therefore need two validities that allow us to put formulas in the correct form.

\begin{lemma}
\label{lemma:AAUL_disjunction}
For every $\phi_1,\phi_2\in \langauml$ and every $\phi_0\in \langpl$, we have
\begin{equation*}\models\langle \AAUL\rangle (\phi_1\vee\phi_2)\leftrightarrow (\langle \AAUL\rangle\phi_1\vee\langle \AAUL\rangle\phi_2)\end{equation*}
and
\begin{equation*}\models\langle\AAUL\rangle (\phi_0\wedge\phi_1)\leftrightarrow (\phi_0\wedge\langle \AAUL\rangle \phi_1).\end{equation*}
\end{lemma}

\begin{proof}
The proof of this lemma is straightforward. 

By the semantic definition, from $M,w \models \dia{\AAUL} (\phi_1\vee\phi_2)$ it follows that $M,w \models \dia{\AAUL} \phi_1$ or that $M,w \models \dia{\AAUL} \phi_2$, and therefore that $M,w \models \langle \AAUL\rangle\phi_1\vee\langle \AAUL\rangle\phi_2$. In the other direction, from $M,w \models \langle \AAUL\rangle\phi_1\vee\langle \AAUL\rangle\phi_2$ it follows that $M,w \models \langle \AAUL\rangle\phi_1$ or that $M,w \models \langle \AAUL\rangle\phi_2$, and therefore, by weakening either formula to the disjunction $\phi_1 \vel \phi_2$, that $M,w \models \langle \AAUL\rangle(\phi_1\vee\phi_2)$.

For the second validity, we first observe that $\dia{\AAUL}\phi_0 \eq \phi_0$ $(*)$ is a validity of AAUML (all our logics are logics of informational change, not of factual change). From $M,w \models \dia{\AAUL} (\phi_0\wedge\phi_1)$ now follows $M,w \models \dia{\AAUL} \phi_0$ and $M,w \models  \dia{\AAUL} \phi_1$, and thus, using $(*)$, that $M,w \models \phi_0$ and $M,w \models  \dia{\AAUL} \phi_1$, and so $M,w \models \phi_0 \et \dia{\AAUL} \phi_1$. For the other direction we observe that, on the assumption of $M,w \models \phi_0\et \dia{\AAUL} \phi_1$, any arrow update executed in $(M,w)$ to make $\phi_1$ true will preserve, by $(*)$, the truth of $\phi_0$, so that $M,w \models \dia{\AAUL} (\phi_0\wedge\phi_1)$.
\end{proof}
It is important and non-trivial to note that the disjunction case can be made constructive. Suppose that we have already synthesized $(U_1,o_1)$ and $(U_2,o_2)$ such that $\models \langle \AAUL\rangle \phi_1\leftrightarrow \langle U_1,o_1\rangle \phi_1$ and $\models \langle \AAUL\rangle \phi_2\leftrightarrow \langle U_2,o_2\rangle \phi_2$. So we have two pointed arrow update models that make $\phi_1$ and $\phi_2$ true whenever possible. This does not, however, immediately give us a \emph{single}-pointed arrow update model $(U,o)$ that guarantees $\phi = \phi_1\vee \phi_2$ whenever possible.\footnote{An alternative technique to synthesize an arrow update for the disjunction is to take the double-pointed direct sum of $(U_1,o_1)$ and $(U_2,o_2)$; its points are $\{o_1,o_2\}$. This method is followed in \cite{hales2013arbitrary}.} In order to find this $(U,o)$, we have to combine $(U_1,o_1)$ and $(U_2,o_2)$. We do this in the following way.

First, we take the set of outcomes of $U$ to be the disjoint union of the sets of outcomes of $U_1$ and $U_2$, plus one extra outcome $o$. Then, we add arrows as follows to $U$.
\begin{quote}For every $(o_1,\psi)\rightarrow_a(o',\psi')$ of $U_1$, add an arrow $(o,\psi\wedge \langle U_1,o_1\rangle \phi_1)\rightarrow_a(o',\psi')$. For every $(o_2,\psi)\rightarrow_a(o',\psi')$ of $U_2$, add an arrow $(o,\psi\wedge \neg \langle U_1,o_1\rangle \phi_1)\rightarrow_a(o',\psi')$. \end{quote}
When executed in a $\langle \AAUL\rangle \phi_1$ world, this arrow update  $(U,o)$ will act as $(U_1,o_1)$, since every such world satisfies $\langle U_1,o_1\rangle \phi_1$ and we added all arrows from $o_1$ with an extra $\langle U_1,o_1 \rangle\phi_1$ precondition. When executed in any $\neg \langle \AAUL\rangle \phi_1$ world, $(U,o)$ acts as $(U_2,o_2)$. As long as either $\langle \AAUL\rangle \phi_1$ or $\langle \AAUL\rangle \phi_2$ holds, we therefore have $\langle U,o\rangle (\phi_1\vee\phi_2)$.

More formally, we have the following lemma.
\begin{lemma}
\label{lemma:AAUL_disjunction_combination}
Let $\phi_1,\phi_2\in \langauml$, and for $i= 1,2$ let $(U_i,o_i)$ be such that $\models \langle \AAUL\rangle \phi_i\leftrightarrow \langle U_i,o_i\rangle\phi_i$, where  $U_i=(O_i,\aufunction_i)$. Furthermore, let $U=(O,\aufunction)$ be given as follows:
\begin{itemize}
	\item $O=\{o\} \uplus O_1 \uplus O_2$,
	\item $\aufunction$ contains exactly the following arrows:
	\begin{enumerate}
		\item $(o',\psi')\rightarrow_a (o'',\psi'')\in \aufunction_i$, for $i=1,2$,
		\item $(o,\psi\wedge \langle U_1,o_1\rangle \phi_1)\rightarrow_a (o',\psi')$ where $(o_1,\psi)\rightarrow_a(o',\psi')\in \aufunction_1$,
		\item $(o,\psi\wedge\neg \langle U_1,o_1\rangle \phi_1)\rightarrow_a (o',\psi')$ where $(o_2,\psi)\rightarrow_a(o',\psi')\in \aufunction_2$.
	\end{enumerate}
\end{itemize}
Then $\models \langle \AAUL\rangle(\phi_1\vee \phi_2) \leftrightarrow \langle U,o\rangle (\phi_1\vee \phi_2)$.
\end{lemma}
\begin{proof}
The right-to-left direction is obvious, because, just as in the proof of Lemma \ref{lemma:AAUL_reduction_1}, source conditions $\psi\wedge \langle U_1,o_1\rangle \phi_1$ and $\psi\wedge\neg \langle U_1,o_1\rangle \phi_1$ can with Corollary~\ref{cor:auml_reduction} be assumed to be equivalent to modal logical formulas.

We now show the left-to-right direction. Suppose therefore that $\M,w\models \langle \AAUL\rangle(\phi_1\vee \phi_2)$. Then, by Lemma~\ref{lemma:AAUL_disjunction}, we have $\M,w\models \langle\AAUL\rangle \phi_1\vee\langle \AAUL\rangle \phi_2$. We continue by a case distinction.

First, suppose that $\M,w\models \langle \AAUL\rangle \phi_1$. Then none of the $(o,\psi\wedge\neg \langle U_1,o_1\rangle \phi_1)\rightarrow_a (o',\psi')$ arrows are applicable in $w$. The $(o,\psi\wedge \langle U_1,o_1\rangle \phi_1)\rightarrow_a (o',\psi')$ arrows, on the other hand, are applicable if and only if $(o_1,\psi)\rightarrow_a(o',\psi')$ applies. Therefore, for every agent $a$, $(w,o)R_a(w',o')$ iff ($wR_aw'$ and $oR_ao'$) iff ($wR_aw'$ and $o_1R_ao')$ iff $(w,o_1)R_a(w',o')$. So there is an $a$-arrow from $(w,o)$ to $(w',o')$ in $M*U$ if and only if there is an $a$-arrow from $(w,o_1)$ to $(w',o')$ in $M*U_1$. Furthermore, because $U$ contains a copy of $U_1$, we have that for every $o',o''\in O_1$ and every $w',w''\in W$ there is an $a$-arrow from $(w',o')$ to $(w'',o'')$ in $M*U$ if and only if there is such an arrow in $M*U_1$. It follows that the relation $\mathfrak{R}_1=\{((w',o'),(w',o'))\mid w'\in W, o'\in O_1\}\cup \{((w,o),(w,o_1))\}$ is a bisimulation between $\M*U,(w,o)$ and $\M*U_1,(w,o_1)$. (It also obviously satisfies {\bf atoms}, as this depends on correspondence of the first argument in each (state, outcome) pair.) By assumption we have $\M,w\models \langle\AAUL\rangle \phi_1$ and therefore $\M,w\models \langle U_1,o_1\rangle \phi_1$, which implies that $\M*U_1,(w,o_1)\models \phi_1$. Because AAUML is invariant under bisimulation it then follows that $\M*U,(w,o)\models \phi_1$ and therefore that $\M,w\models\langle U,o\rangle \phi_1$, so that we also have $\M,w\models\langle U,o\rangle (\phi_1\vee\phi_2)$.

Suppose then that $\M,w\not \models \langle \AAUL\rangle \phi_1$. Then we must have $\M,w\models \langle \AAUL\rangle \phi_2$. In this case, the $(o,\psi\wedge \langle U_1,o_1\rangle \phi_1)\rightarrow_a (o',\psi')$ arrows are inapplicable while the $(o,\psi\wedge\neg \langle U_1,o_1\rangle \phi_1)\rightarrow_a (o',\psi')$ arrows apply if and only if $(o_2,\psi)\rightarrow_a(o',\psi')$ does. Similarly to the previous case, the relation $\mathfrak{R}_2=\{((w',o'),(w',o'))\mid w'\in W, o'\in O_2\}\cup \{((w,o),(w,o_2))\}$ is therefore a bisimulation between $\M*U,(w,o)$ and $\M*U_2,(w,o_2)$.
From the assumption that $\M,w\models \langle \AAUL\rangle \phi_2$ it follows that $\M,w\models \langle U_2,o_2\rangle \phi_2$ and therefore that $M*U_2,(w,o_2)\models \phi_2$. Because AAUML is invariant under bisimulation we then obtain $M*U,(w,o)\models \phi_2$ and therefore $M,w\models \langle U,o\rangle \phi_2$ and also $M,w\models \langle U,o\rangle (\phi_1\vee \phi_2)$.

In either case, we have $M,w\models \langle U,o\rangle (\phi_1\vee \phi_2)$, as was to be shown.
\end{proof}

\subsection{Reduction}

Using the fact that $[U,o]$ commutes with $\neg$, distributes over $\wedge$ and, in a convoluted but not very complicated way, commutes with $\square_a$, Corollary~\ref{cor:auml_reduction} showed that every $\langauml$ formula is equivalent to a $\langml$ formula. In order to be able to finally perform synthesis, it remains to show that every $\lang$ formula is also equivalent to a $\langml$ formula. This is what we will show in this section.

In Section~\ref{sec:red_AAUL} we showed that $\langle \AAUL\rangle$ commutes, in a very complicated way, with boolean combinations of basic modal formulas. Also, like $[U,o]$, a $\langle \AAUL\rangle$ operator disappears once it encounters a propositional atom. From this we can show that, if $\phi$ is a formula of  basic modal logic, then we can transform $\langle \AAUL\rangle \phi$ into an equivalent formula $\phi'$ of  basic modal logic.

The proof that $\langle \AAUL\rangle \phi$ can be transformed into $\phi'$ is by induction on the order $\succ$ given by: $\phi_1\succ \phi_2$ if and only if (i) $\phi_2$ is a strict subformula of $\phi_1$ or (ii) $\phi_2$ has a (strictly) lower modal depth than $\phi_1$. Let us first show that this order is well-founded.

\begin{lemma}
\label{lemma:prec_wf}
The relation $\succ$ is well-founded.
\end{lemma}
\begin{proof}
This follows from the fact that both the $>$ relation on the natural numbers and the ``strict subformula of'' relation are well-founded.

Suppose towards a contradiction that there is an infinitely descending chain $\phi_0 \succ \phi_1\succ \cdots$ of formulas. 
So for every $i\in \mathbb{N}$ we have either (i) $\phi_{i+1}$ is a strict subformula of $\phi_i$ or (ii) $d(\phi_i)>d(\phi_{i+1})$. The modal depth of a formula is at least as large as that of its subformulas, so in either case we have $d(\phi_i)\geq d(\phi_{i+1})$. Because the depth of a formula is a natural number and $>$ is well-founded on $\mathbb{N}$, there can be only finitely many $i$ such that $d(\phi_{i})>d(\phi_{i+1})$. So there is some $k\in \mathbb{N}$ such that for all $m\geq k$ we have $d(\phi_m)=d(\phi_k)$.

It follows that for every $m\geq k$, the formula $\phi_{m+1}$ must be a strict subformula of $\phi_{m}$. This is impossible, however, since the ``strict subformula of'' relation is well-founded. We have arrived at a contradiction, so $\succ$ does not have an infinitely descending sequence and is therefore well-founded.
\end{proof}
Now, we can prove the reduction from $\langle \AAUL\rangle\phi$ to $\phi'$.

\begin{lemma}
\label{lemma:aauml_reduction}
For every $\phi \in \langml$, there is a formula $\phi'\in \langml$ such that $\models \langle \AAUL\rangle \phi\leftrightarrow \phi'$.
\end{lemma}
\begin{proof}
Every formula is equivalent to a formula in DNNF (See Sections~\ref{sec.syntax} and \ref{sec.semantics}), so we can assume without loss of generality that $\phi$ is in DNNF. We can also assume without loss of generality that every conjunction in $\phi$ contains exactly one conjunct of the form $\square_a\chi$ for every agent $a$. We now proceed by induction on the order $\succ$, which, by the previous lemma, is well-founded. Note that the minimal elements with respect to $\succ$ are the formulas that are of depth 0 and have no strict subformulas, i.e., the atoms.


Suppose therefore, as a base case in our induction, that $\phi$ is an atom. Then $\models \langle \AAUL\rangle \phi \leftrightarrow \phi$, so the lemma holds with $\phi'=\phi$. 
Assume now as induction hypothesis that the lemma holds for every $\phi'$ such that $\phi \succ \phi'$. We then continue by case distinction. 

First, suppose that $\phi = \phi_1\vee \phi_2$. By Lemma~\ref{lemma:AAUL_disjunction} we have $\models \langle \AAUL\rangle (\phi_1\vee \phi_2)\leftrightarrow (\langle \AAUL\rangle \phi_1\vee \langle \AAUL\rangle \phi_2)$. Furthermore, $\phi\succ \phi_1$ and $\phi\succ\phi_2$, so by the induction hypothesis there are $\phi_1',\phi_2'\in \langml$ such that $\models \langle \AAUL\rangle\phi_1\leftrightarrow\phi_1'$ and $\models \langle \AAUL\rangle \phi_2\leftrightarrow\phi_2'$. It follows that $\models \langle \AAUL\rangle \phi\leftrightarrow (\phi_1'\vee\phi_2') $.

Secondly, suppose that $\phi = \phi_0\wedge \phi_1$, where $\phi_0\in \langpl$. Then, by Lemma~\ref{lemma:AAUL_disjunction}, we have $\models \langle \AAUL\rangle (\phi_0\wedge \phi') \leftrightarrow (\phi_0\wedge \langle \AAUL\rangle\phi_1)$. Furthermore, $\phi\succ \phi_1$, so by the induction hypothesis there is a $\phi_1'$ such that $\models \langle \AAUL\rangle \phi_1\leftrightarrow\phi_1'$. It follows that $\models \langle \AAUL\rangle\phi\leftrightarrow (\phi_0\wedge \phi_1')$.

Finally, suppose that $\phi$ is a conjunction without a propositional conjunct. Then for every $a\in \Agents$ there are a finite set $\Phi_a\subseteq \langml$ and a formula $\psi_a\in \langml$ such that $\phi = \bigwedge_{a\in\Agents}(\bigwedge_{\phi_a\in\Phi_a}\lozenge_a\phi_a\wedge\Box_a\psi_a)$. By Lemma~\ref{lemma:AAUL_main_reduction}, we have $\models \langle \AAUL\rangle \phi \leftrightarrow \bigwedge_{a\in\Agents}\bigwedge_{\phi_a\in \Phi_a}\lozenge_a\langle \AAUL\rangle (\phi_a\wedge \psi_a)$. For every $a\in \Agents$ and $\phi_a\in\Phi_a$, the modal depth of $\phi_a\et\psi_a$ is strictly lower than that of $\phi$, so $\phi\succ \phi_a\et\psi_a$. Therefore, by the induction hypothesis, there is a formula $\phi_a'\in \langml$ such that $\models {\langle\AAUL\rangle} (\phi_a\wedge\psi_a)\leftrightarrow \phi_a'$. Let $\Phi_a'$ be the set of such $\phi_a'$. It follows that $\models \langle \AAUL\rangle \phi\leftrightarrow\bigwedge_{a\in\Agents}\bigwedge_{\phi_a'\in \Phi_a'}\lozenge_a\phi_a'$. 

This completes the induction step and thereby the proof.
\end{proof}

We now have the following theorem.
\begin{theorem}
\label{thm:expressivity}
For every $\phi\in\lang$ there is a formula $\phi'\in\langml$ such that $\models \phi \leftrightarrow \phi'$.
\end{theorem}
\begin{proof}
We now use Lemma \ref{lemma:aauml_reduction} (in dual form) that \begin{quote} for every $\phi \in \langml$ there is a $\phi'\in \langml$ such that $\models [\AAUL] \phi\leftrightarrow \phi'$, \end{quote} and we also use the consequence of Lemma~\ref{lemma:aauml_reduction} that \begin{quote} for every $\phi \in \langml$ and arrow update $(U,o)$ with source and target conditions in $\langml$ there is a $\phi'\in \langml$ such that $\models [U,o] \phi\leftrightarrow \phi'$,\end{quote} in order to successively eliminate the innermost $[U,o]$ or $[\AAUL]$ operators of any given formula of AAUML. We can thus transform this formula into an equivalent formula of modal logic.
\end{proof}
It follows that in Lemmas~\ref{lemma:reduction_simple_2a}--\ref{lemma:AAUL_disjunction_combination} the restriction to the sublanguage $\langauml$ is unnecessary; the lemmas also hold for the full language $\lang$. In this form we will also later use these validities as axioms in the axiomatization. 
For the next section, it is important to keep in mind that the reduction axioms not only guarantee the existence of $\phi'$, but also enable us to find it.

Theorem \ref{thm:expressivity} also allows us to prove a claim that we made in Section~\ref{sec.semantics}. There, we defined $M,s\models[\AAUL]\phi$ by
\begin{quote}
$M,s\models [\AAUL]\phi$ iff $(M,s\models [U,o]\phi$ for every arrow update  $(U,o)$ that has source and target conditions only in $\langml$).
\end{quote}
Now that we have shown that every formula of $\lang$ is equivalent to a formula of $\langml$, it follows immediately that the requirement of the source and target conditions being in $\langml$ is unnecessary.
\begin{proposition}
\label{prop:fully_arbitrary}
For every $\phi\in\lang$ and every pointed relational model $(M,s)$, we have
\begin{center}
$M,s\models [\AAUL]\phi$ \ \ \ iff \ \ \ $M,s\models [U,o]\phi$ for every arrow update  $(U,o)$.
\end{center}
\end{proposition}

\subsection{Synthesis}
Recall that our goal, when performing synthesis, is to find, for given $\phi\in \mathcal{L}$, an arrow update  $(U_\phi,o_\phi)$ such that $\models \langle \AAUL\rangle \phi\leftrightarrow \langle U_\phi,o_\phi\rangle \phi$. Using Theorem~\ref{thm:expressivity}, we can transform $\phi$ into an equivalent formula of modal logic. Then, using the procedure outlined in Section~\ref{sec:red_AAUL}, we can find $(U_\phi,o_\phi)$. The procedure is found in detail in Table \ref{table.proc}.

\begin{table}[h]
{\setlength{\extrarowheight}{2pt}
\begin{tabular}{ll}
\hline
\multicolumn{2}{c}{Procedure $\mathit{Synth}(\phi)$}\vspace{2pt}\\
\hline
\hline
 & Input: $\phi \in \mathcal{L}$.\\
 & Output: $(U_\phi,o_\phi)$ such that $\models \langle \AAUL\rangle \phi \leftrightarrow \langle U_\phi,o_\phi\rangle \phi$.\\
 \hline
 \hline
1. & If $\phi\not \in \langml$, then use the reduction axioms to find a formula $\phi_\mathit{modal}\in \langml$ \\ &such that $\models \phi\leftrightarrow \phi_\mathit{modal}$, and return $\mathit{Synth}(\phi_\mathit{modal})$. Otherwise, continue.\\
2. & If $\phi$ is not in disjunction normal form, compute the DNNF $\phi_\mathit{DNNF}$ of $\phi$ and \\ &return $\mathit{Synth}(\phi_\mathit{DNNF})$. Otherwise, continue.\\
3. & If $\phi = \phi_1\vee \phi_2$, then compute $\mathit{Synth}(\phi_1)$ and $\mathit{Synth}(\phi_2)$, and combine \\ &the two arrow update models as in Lemma~\ref{lemma:AAUL_disjunction_combination}. Return the combined \\ & arrow update model.\\
4. & If $\phi$ is not a disjunction, then since it is in DNNF it must be a conjunction, \\&where each conjunct is (i) a literal, (ii) of the form $\lozenge_a\psi$, or (iii) of the form $\square_a\chi$. \\ & Assume w.l.o.g. that for every $a$ there is exactly one conjunct $\square_a\chi_a$.\\ & For every $\lozenge_a\psi$, compute $\mathit{Synth}(\psi\wedge\chi_a)$. Then use Lemma~\ref{lemma:AAUL_reduction_1} to combine\\ & the arrow update models, and return the result. If there are no $\lozenge_a$ operators for any\\& agent $a$, return the trivial arrow update model with one outcome and no arrows.\\
\hline
\end{tabular}}
\caption{Synthesis procedure}
\label{table.proc}
\end{table}

Our arrow update synthesis was motivated by Hales' already mentioned action model synthesis published in \cite{hales2013arbitrary}. See also the next Section \ref{sec.axiomatization}. Subsequently Hales {\em et al.} investigated these matters in \cite{frenchetal:2014} and in (his Ph.D.\ thesis) \cite{hales:2016}. Alternatively, action model synthesis by way of a dedicated action language was employed by Aucher in \cite{aucher.jancl:2011}.

\subsection{Example}
In order to gain better understanding of $\mathit{Synth}(\phi)$, let us consider an example. 
Suppose $\phi = \lozenge_a\square_bp\wedge\lozenge_b (\lozenge_aq \vel \lozenge_ar)\wedge\square_bp)$. We want to perform synthesis for this $\phi$.

\medskip

\noindent {\bf Goal: find $\mathit{Synth}(\lozenge_a\square_bp\wedge\lozenge_b(\lozenge_aq\vee\lozenge_ar)\wedge\square_bp)$.}


Because $\phi\in \langml$, $\phi$ is in DNNF and $\phi$ is not a disjunction, we continue past steps 1, 2 and 3. In step 4, it is assumed that for every agent there is exactly one $\square$ conjunct. This means we need to add a trivial conjunct $\square_a\top$. 

Of the conjuncts of $\phi$, two have $\lozenge$ as primary operator. So we need to perform synthesis for two more formulas; because of $\lozenge_a\square_bp$ and $\square_a\top$ we need to find $\mathit{Synth}(\square_bp\wedge \top)$, and because of $\lozenge_b(\lozenge_aq\vee\lozenge_ar)$ and $\square_bp$ we need to find $\mathit{Synth}((\lozenge_aq\vee\lozenge_ar)\wedge p)$.
\begin{adjustwidth}{1cm}{}
	{\bf Subgoal 1: find $\mathit{Synth}(\square_bp\wedge \top)$.}\\
	Since we are doing synthesis for a conjunction, we continue in steps 1, 2 and 3. Because there are no $\lozenge$ operators in $\square_bp\wedge\top$, we return the trivial arrow update model in step 4.\\
	{\bf Subgoal 2: find $\mathit{Synth}((\lozenge_aq\vee\lozenge_ar)\wedge p)$.}
	In step 2, we need to put the formula in DNNF. We therefore continue with $(\lozenge_aq\wedge p) \vee (\lozenge_ar\wedge p)$. In step 3 we are then instructed to perform synthesis for the two disjuncts.
	\begin{adjustwidth}{1cm}{}
		{\bf Sub-subgoal 2.1: find $\mathit{Synth}(\lozenge_aq\wedge p)$.}\\
		We continue up to step 4. There, we first add a a trivial conjunct $\square_a\top$. Then, we are instructed to find $\mathit{Synth}(q\wedge \top)$.
		\begin{adjustwidth}{1cm}{}
			{\bf Sub-sub-subgoal 2.1.1: find $\mathit{Synth}(q\wedge \top)$.}\\
			We proceed to step 4. There, since there are no $\lozenge$ operators in $q\wedge \top$, we return the trivial arrow update model $(U_0,o_0)$.
		\end{adjustwidth}
		Now, in order to find $\mathit{Synth}(\lozenge_aq\wedge p)$, we take the trivial arrow update model found in sub-sub-subgoal 2.1.1, and add one extra outcome. Then, we connect this extra outcome to the trivial model by a $\top \imp_a\langle U_0,o_0\rangle\top$ arrow. The source condition of this arrow is $\top$ because step 4 uses the construction from Lemma~\ref{lemma:AAUL_reduction_1}, and that construction always gives precondition $\top$. The arrow is for agent $a$, because we started with a $\lozenge_a$ operator. Finally, the target condition is $\langle U_0,o_0\rangle \top$ because the arrow update that the arrow points to is $(U_0,o_0)$ and the $\square_a$ conjunct was $\square_a\top$. 	
		In other words, we obtain the arrow update depicted in Figure~\ref{fig.synthesisexample}.a, where the framed state indicates the designated outcome.
		
\noindent		{\bf Sub-subgoal 2.2: find $\mathit{Synth}(\lozenge_ar\wedge p)$.}\\
		Replacing the $q$ of $\lozenge_aq\wedge p$ for an $r$ does not change the arrow update model that we end up with. So in this sub-subgoal we find the same model as in sub-subgoal 2.1.
	\end{adjustwidth}
	Now, in order to find $\mathit{Synth}((\lozenge_aq\wedge p) \vee (\lozenge_ar\wedge p))$, we need to combine the models found in sub-subgoals 2.1 and 2.2. Since we are working with a disjunction, we combine them as described in Lemma~\ref{lemma:AAUL_disjunction_combination}. That, is, we take copies of the two (identical) models and add one extra outcome. Then, we add two more arrows: every world that is reachable from the origin world of the model from sub-subgoal 2.1 by $\psi_1\imp_a\psi_2$, becomes reachable from the extra world by a $\psi_1\wedge\langle \AAUL\rangle (\lozenge_aq\wedge p) \imp_a\psi_2$ arrow. Likewise, every world reachable by $\psi_1 \imp_a \psi_2$ from the origin of the model from sub-subgoal 2.2 becomes reachable from the extra world by $\psi_1\wedge\neg\langle \AAUL\rangle (\lozenge_aq\wedge p) \imp_a\psi_2$. We now get the model depicted in Figure~\ref{fig.synthesisexample}.b.
\end{adjustwidth}

\begin{figure}
\center
		\fbox{
		\begin{tikzpicture}[font=\footnotesize]
			\node[draw, rectangle,minimum size=4pt,thick] (o) at (0,0) {};
			\fill (0,0) circle (0.04) node {};
			\fill (6,0) circle (0.04) node (o1) {};
			
			\draw[->] (o) -- node[above,sloped]{$\top \ \ _a \ \ \langle U_0,o_0\rangle \top$} (o1);
		\end{tikzpicture}
} (a)

		\fbox{
		\begin{tikzpicture}[font=\footnotesize]
		\fill (0,5) circle (0.04) node (o0) {};
		\fill (6,5) circle (0.04) node (o1) {};
		
		\fill (0,0) circle (0.04) node (o2) {};
		\fill (6,0) circle (0.04) node (o3) {};
		
		\fill (0,2.5) circle (0.04) node (o) {};
		\node[draw, rectangle,minimum size=4pt,thick] at (0,2.5) {};
		\draw[->] (o0)-- node[above,sloped]{$\top \ \ _a \ \ \langle U_0,o_0\rangle \top$} (o1);
		\draw[->] (o2) -- node[below,sloped]{$\top \ \ _a \ \ \langle U_0,o_0\rangle \top$} (o3);
		\draw[->] (o) -- node[above,sloped,pos=0.45]{$\top\wedge \langle \AAUL\rangle(\lozenge_aq\wedge p) \ \ _a \ \ \langle U_0,o_0\rangle \top$} (o1);
		\draw[->] (o) -- node[below,sloped,pos=0.45]{$\top\wedge \neg\langle \AAUL\rangle(\lozenge_aq\wedge p) \ \ _a \ \ \langle U_0,o_0\rangle \top$} (o3);
	\end{tikzpicture}
} (b)
		\fbox{
\begin{tikzpicture}[font=\footnotesize]
	\fill (-5,3.75) circle (0.04) node (o-) {};
	\node[draw, rectangle,minimum size=4pt,thick] at (-5,3.75) {};

	\fill (0,7.5) circle (0.04) node (o4) {};

	\fill (0,5) circle (0.04) node (o0) {};
	\fill (6,5) circle (0.04) node (o1) {};
		
	\fill (0,0) circle (0.04) node (o2) {};
	\fill (6,0) circle (0.04) node (o3) {};
	
	\fill (0,2.5) circle (0.04) node (o) {};
		
	\draw[->] (o0)-- node[above,sloped]{$\top \ \ _a \ \ \langle U_0,o_0\rangle \top$} (o1);
	\draw[->] (o2) -- node[below,sloped]{$\top \ \ _a \ \ \langle U_0,o_0\rangle \top$} (o3);
	\draw[->] (o) -- node[above,sloped,pos=0.45]{$\top\wedge \langle \AAUL\rangle(\lozenge_aq\wedge p) \ \ _a \ \ \langle U_0,o_0\rangle \top$} (o1);
	\draw[->] (o) -- node[below,sloped,pos=0.45]{$\top\wedge \neg\langle \AAUL\rangle(\lozenge_aq\wedge p) \ \ _a \ \ \langle U_0,o_0\rangle \top$} (o3);
	
	\draw[->] (o-) -- node[above,sloped]{$\top \ \ _a \ \ \langle U_0,o_0\rangle\top$} (o4);
	\draw[->] (o-) -- node[below,sloped]{$\top \ \ _b \ \ \langle U_1,o_1\rangle p$} (o);
\end{tikzpicture}
} (c) 

		\fbox{
	\begin{tikzpicture}
		\fill (0,0) circle (0.04) node (o) {};
		\fill (3,-3) circle (0.04) node (o1) {};
		\fill (7,-3) circle (0.04) node (o2) {};
		
		\node[draw, rectangle,minimum size=4pt,thick] at (0,0) {};
		
		\draw[->] (o) -- node[below,sloped]{$\top \ \ _b \ \ p$} (o1);
		\draw[->] (o) -- node[above,sloped]{$\top \ \ _a \ \ \top$} (o2);
		\draw[->] (o1) -- node[below,sloped]{$\top \ \ _a \ \ \top$} (o2);
	\end{tikzpicture}
} (d)
\caption{Different stages in the synthesis of $ \lozenge_a\square_bp\wedge\lozenge_b (\lozenge_aq\vel\lozenge_ar)\wedge\square_bp)$}
\label{fig.synthesisexample}
\end{figure}
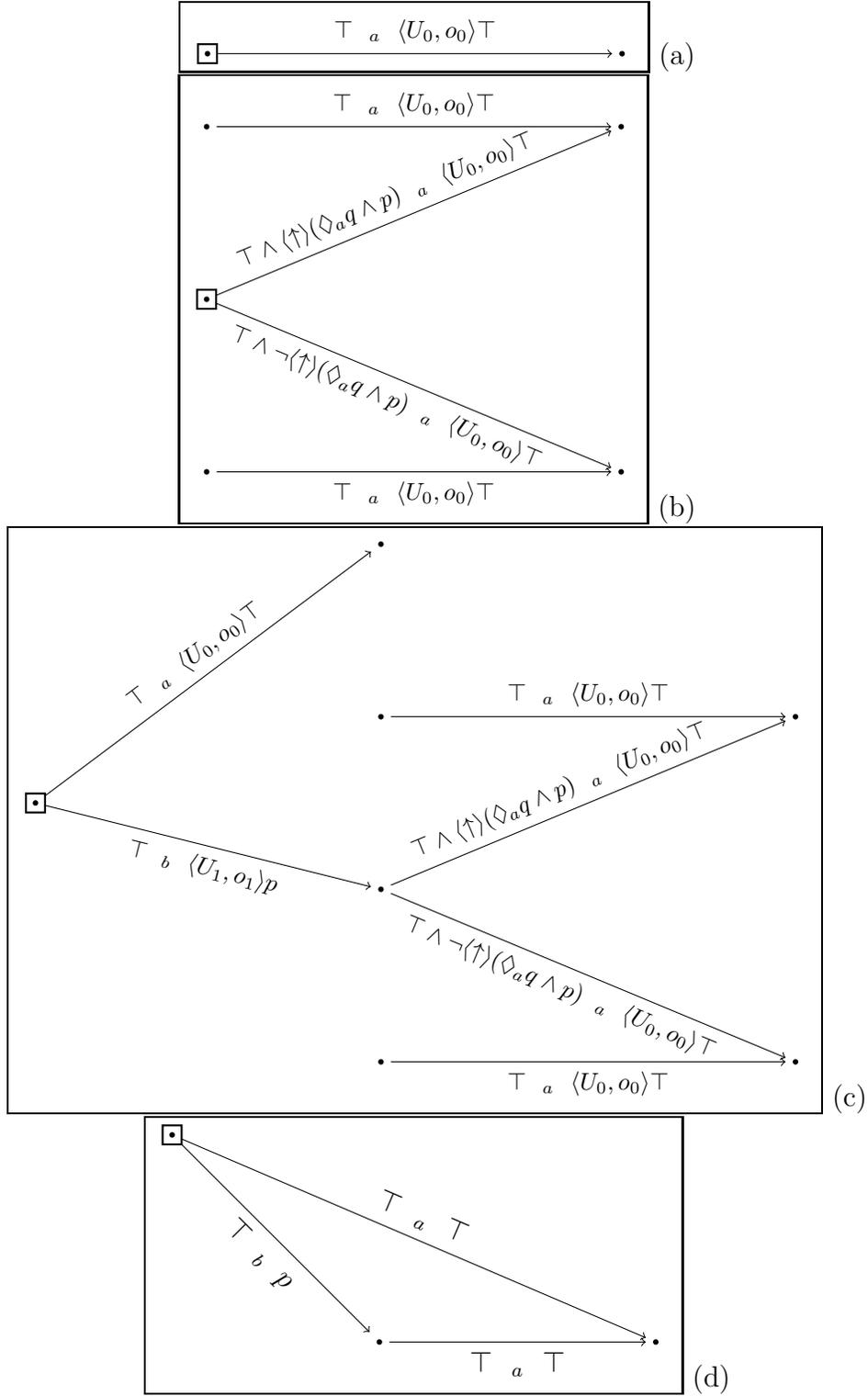

\noindent Now, all that is left to do is to combine the arrow update models found in subgoals 1 and 2. The model we obtain is depicted in Figure~\ref{fig.synthesisexample}.c, where $(U_1,o_1)$ is the model that we found in subgoal 2. The root of the model is the leftmost outcome. Note that the depth (i.e., the maximum path length) of this arrow update model is 2, just like the depth (i.e., the maximum number of nested $\square$ or $\lozenge$ operators) of $\phi$. In general, the depth of the synthesized arrow update model is bounded by that of the formula for which synthesis is performed.

Also, note that the arrow update model that we obtained can quite easily be modified to obtain a smaller model that is still sufficient. In particular:
\begin{itemize}
	\item the two outcomes that are not reachable from the root can be eliminated,
	\item the formulas $\langle U_0,o_0\rangle\top$, $\langle U_1,o_1\rangle p$, $\top \wedge \langle \AAUL\rangle (\lozenge_aq\wedge p)$ and $\top\wedge\neg \langle\AAUL\rangle(\lozenge_aq\wedge p)$ can be replaced by the equivalent formulas $\top$, $p$, $\lozenge_aq\wedge p$ and $\neg (\lozenge_aq\wedge p)$, respectively,
	\item the three leaf outcomes can be merged into one, 
	\item and $\lozenge_aq\wedge p \imp_a \top$ and $\neg(\lozenge_a q\wedge p) \imp_a \top$ can be merged into one $\top\imp_a\top$ arrow.
\end{itemize}
With these optimizations, we get the more aesthetically pleasing arrow update depicted in Figure~\ref{fig.synthesisexample}.d.

\section{Axiomatization} \label{sec.axiomatization}

Using the reduction axioms introduced before, we can find an axiomatization for AAUML. Let {\bf AAUML} be the axiomatization shown in Table~\ref{table.axiomatization}. In this section we show that the axiomatization {\bf AAUML} is sound and complete, and we give some derivable (well-known) axiom schemata. 

\begin{table}[h]
\[ \begin{array}{ll}
{\bf Prop} & \text{all tautologies of propositional logic}\\
{\bf K} & \Box_\agent(\phi\imp\psi)\imp (\Box_\agent\phi\imp \Box_\agent\psi) \\
{\bf U1} & [U,o]p \eq p \\
{\bf U2} & [U,o]\neg\phi \eq \neg [U,o]\phi \\
{\bf U3} & [U,o](\phi\et\psi) \eq ([U,o]\phi \et [U,o]\psi) \\
{\bf U4} & [U,o]\Box_\agent \phi \eq \Et_{(o,\psi) \imp_a (o',\psi')} (\psi \imp \Box_a (\psi' \imp [U,o']\phi)) \\
{\bf A1} & \langle \AAUL\rangle \phi_0 \leftrightarrow \phi_0$ \hspace{6cm} \hfill where $\phi_0\in \langpl\\
{\bf A2} & \langle \AAUL \rangle (\phi \vee \psi)\leftrightarrow (\langle \AAUL\rangle \phi \vee \langle \AAUL\rangle \psi)\\
{\bf A3} & \langle \AAUL \rangle (\phi_0 \wedge \phi)\leftrightarrow (\phi_0\wedge \langle \AAUL\rangle\phi)$ \hfill where $\phi_0\in \langpl\\
{\bf A4} & \langle \AAUL\rangle\bigwedge_{a\in A}(\bigwedge_{\phi_a\in \Phi_a}\lozenge_a\phi_a\wedge \square_a\psi_a)\leftrightarrow  \bigwedge_{a\in A}\bigwedge_{\phi_a\in \Phi_a}\lozenge_a\langle\AAUL\rangle(\phi_a\wedge\psi_a) \hspace{1cm} \ \\
{\bf MP} & \text{from } \phi\imp\psi \text{ and } \phi \text{ infer } \psi\\
{\bf NecK} & \text{from } \phi \text{ infer } \Box_\agent\phi\\
{\bf RE} & \text{from } \chi\eq\psi \text{ infer } \phi[\chi/\atom] \eq \phi[\psi/\atom]
\end{array} \]
\caption{The axiomatization {\bf AAUML} of the logic AAUML}
\label{table.axiomatization}
\end{table}

\begin{lemma}
Axiomatization {\bf AAUML} is sound for the logic AAUML.
\end{lemma}
\begin{proof}
{\bf Prop}, {\bf K}, {\bf MP}, {\bf NecK}, {\bf RE} are known from modal logic, {\bf U1}---{\bf U4} were demonstrated in Section~\ref{sec:red_U} and originate in \cite{kooirenne}, {\bf A1}---{\bf A4} were shown to be valid in Section~\ref{sec:red_AAUL}.
\end{proof}
It is important to note that {\bf U1}---{\bf U4} and {\bf A1}---{\bf A4} are so-called \emph{reduction axioms} for the operators $[U,o]$ and $\langle \AAUL\rangle$, respectively, as mentioned in the previous section. This means that they are equivalences, where the formula inside the scope of the $[U,o]$ or $\langle \AAUL\rangle$ operator on the left-hand side is more complex than the formulas inside the scope of that operator on the right-hand side.

The derivation rule {\bf RE} is important as our reductions are inside-out, not outside-in. Without it, for example, the validity $[U,o][U,o](p\vel\neg p)$ would not be derivable.

The axioms {\bf A1}---{\bf A4} could just as well have been formulated with the $[\AAUL]$ dual of the modality $\dia{\AAUL}$, e.g., $\mathbf{A2'} \ [\AAUL](\phi\et\psi) \eq ([\AAUL]\phi \et [\AAUL]\psi)$. We prefer the $\dia{\AAUL}$ versions as they match our usage of these axioms in synthesis. 
Further note that there is no reduction of shape $\dia{\AAUL}\neg\phi\eq\dots$. We assume that subformulas bound by $\dia{\AAUL}$ are first transformed into disjunctive negation normal form before a further reduction can take place (and again, for this, the derivation rule {\bf RE} is essential).

\begin{lemma}
Axiomatization {\bf AAUML} is complete for the logic AAUML.
\end{lemma}
\begin{proof}
Let $\phi\in\lang$ be valid. Using an induction argument, we can eliminate all $[U,o]$ and $\langle \AAUL\rangle$ operators from it: $\phi$ must be provably equivalent to a formula $\phi'\in\langml$. As $\phi'$ must also be valid (Theorem \ref{thm:expressivity}), $\phi'$ is provable in modal logic. From the provable equivalence between $\phi$ and $\phi'$ and the derivation of $\phi'$ we construct a derivation of $\phi$ in {\bf AAUML}.
\end{proof}

We have now shown that:
\begin{theorem}
Axiomatization {\bf AAUML} is sound and complete for the logic AAUML.
\end{theorem}

In the proof system {\bf AAUML}, we do not have necessitation for the $[U,o]$ and $[\AAUL]$ operators. Such necessitation rules are derivable, however.

\begin{proposition}
The following two rules are derivable in {\bf AAUML}. \begin{itemize} \item {\bf NecU}: from $\phi$ infer $[U,o]\phi$; \item {\bf NecA}: from $\phi$ infer $[\AAUL]\phi$. \end{itemize}
\end{proposition}
\begin{proof}
First, note that the axiom 
\begin{center}
\begin{tabular}{ll}
{\bf U1$'$} & $[U,o]\phi_0\leftrightarrow \phi_0$ \hspace{2cm} where $\phi_0\in \langpl$
\end{tabular}
\end{center}
is derivable, using {\bf Prop}, {\bf U1}--{\bf U3} and {\bf MP}. It is also convenient to use a variant of {\bf MP} directly on bi-implications, instead of first converting the bi-implication to a single implication.
\begin{center}
\begin{tabular}{ll}
{\bf MP$'$} & from $\phi \leftrightarrow \psi$ and $\phi$ infer $\psi$, and from $\phi \leftrightarrow \psi$ and $\psi$ infer $\phi$.
\end{tabular}
\end{center}
This {\bf MP$'$} is, of course, also easily derived. Using {\bf U1$'$} and {\bf MP$'$}, we can derive {\bf NecU} in a reasonable number of steps:
\begin{center}
\begin{tabular}{lll}
1.& $\phi$ & premise\\
2.& $\phi \rightarrow (\phi \leftrightarrow \top)$ &  {\bf Prop}\\
3.& $\phi \leftrightarrow \top$ & {\bf MP}(1,2)\\
4.& $[U,o]\top \leftrightarrow \top$ & {\bf U1$'$}\\
5.& $\top$ & {\bf Prop}\\
6.& $[U,o]\top$ & {\bf MP$'$}(5,4)\\
7.& $[U,o]\top \leftrightarrow [U,o]\phi$ & {\bf RE}(3)\\
8.& $[U,o]\phi$&{\bf MP$'$}(6,7) 
\end{tabular}
\end{center}
A derivation of {\bf NecA} can be found in a similar way. Here, too, it is convenient to first derive an auxiliary axiom.
\begin{center}
\begin{tabular}{ll}
{\bf A1$'$} & $[\AAUL]\phi_0\leftrightarrow \phi_0$ \hspace{2cm} where $\phi_0\in \langpl$
\end{tabular}
\end{center}
This $[\AAUL]$-version of {\bf A1} is of course derivable. We can then derive {\bf NecA} analogously to how we derived {\bf NecU}, with the application of {\bf U1$'$} replaced by {\bf A1$'$}.
\end{proof}

The axiomatization {\bf AAUML} is inspired by the somewhat similar axiomatization by Hales of \emph{arbitrary action model logic} \cite{hales2013arbitrary}, although the shape of some axioms and rules is rather different. Axioms in Hales similar to our {\bf U1}---{\bf U4}, for the reduction for action models, are of course taken from action model logic instead. Axioms in Hales similar to our {\bf AI}---{\bf A4}, that are used for the reduction of the quantifier, are taken from refinement modal logic instead. There is nothing `of course' about the latter: Hales' mixture of AML and RML was, we think, a quite original move. One can also move in the other direction: our axioms {\bf U1}---{\bf U4} allow an alternative axiomatization of refinement modal logic. However, as this seems out of scope, these results are not presented here.

\section{Update expressivity} \label{sec.updateex}

\subsection{Expressivity}

Recall that we are considering the basic modal logic ML and the update logics PAL, APAL, AML, AAML, AUL, AAUL, AUML, AAUML, and RML, as shown in Figure~\ref{fig:logic_relations} on page \pageref{fig:logic_relations}. One natural thing to do with such related logics that are all interpreted on a similar class of structures is to compare their power to relate or to distinguish structures from that class. The most straightforward way to make such a comparison is to compare their \emph{expressivity}. 

Formally, given languages $\lang_1$ and $\lang_2$ interpreted on certain class of models $\mathcal{M}$, a language $\mathcal{L}_1$ is \emph{at least as expressive} as a language $\mathcal{L}_2$ if for every formula of $\mathcal{L}_2$ there is an equivalent formula of $\mathcal{L}_1$. Having {\em equal expressivity} or {\em higher expressivity} (by which we always mean strictly higher expressivity) can be defined from the ``at least as expressive'' relation in the usual way. If neither language is at least as expressive as the other, we say that they are {\em incomparable} in expressivity. Informally, if $\mathsf{L}_1$ and $\mathsf{L}_2$ are the logics for languages $\lang_1$ resp.\ $\lang_2$ interpreted on $\mathcal{M}$, we will also use the terminology for expressivity to compare these logics.

In \cite{plaza:1989} that introduced PAL it was shown that PAL is equally expressive as ML  (on the class $\mathcal{S}5$ of relational models, but this does not matter for the reduction), and in \cite{baltagetal:1998} that introduced AML it was also shown that AML is equally expressive as ML. It is trivial to show that PAL and AML are at least as expressive as ML, as they extend the logical language. That every formula of PAL or AML is equivalent to a formula in ML, was shown by reduction axioms and rules. Similarly, AUL \cite{kooietal:2011}, AUML \cite{kooirenne} and AAML \cite{hales2013arbitrary} were shown to be equally expressive as ML, and therefore also equally expressive as PAL and AML. In \cite{bozzellietal.inf:2014} it was shown that RML is equally expressive as ML. Here, in Theorem~\ref{thm:expressivity} in Section~\ref{sec.synthesis}, we showed that AAUML is also  equally expressive as ML.

The two remaining logics are APAL and AAUL. The logic APAL was shown to be more expressive than ML in \cite{balbianietal:2008} and AAUL was shown to be more expressive than ML in \cite{hvdetal.aij:2017}, wherein it was also shown that APAL and AAUL are incomparable. This means that the expressivity landscape is as shown in Figure~\ref{fig:expressivity}.

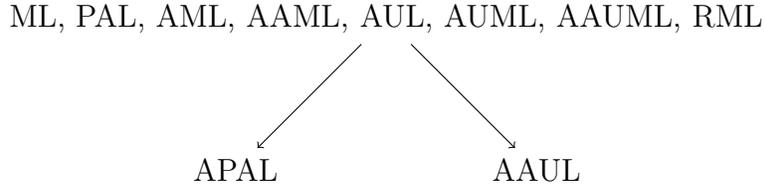
\begin{figure}[h]
\center
\begin{tikzpicture}
\node (n0) at (0,0) {ML, PAL, AML, AAML, AUL, AUML, AAUML, RML};
\node (n1) at (-2,-2) {APAL};
\node (n2) at (2,-2) {AAUL};

\draw[->] (n0) -- (n1);
\draw[->] (n0) -- (n2);
\end{tikzpicture}
\caption{The relative expressivity of the update logics discussed in the paper. Arrows indicate increasing expressivity. Absence of arrows indicates incomparable expressivity.}
\label{fig:expressivity}
\end{figure}

\subsection{Update expressivity hierarchy}

There is something a bit strange about this comparison, however. Although AML and PAL have the same expressivity, AML is clearly in some sense more powerful, since action models represent a far larger class of updates than public announcements. In order to capture the sense in which AML is more powerful than PAL, we use the term \emph{update expressivity}. This concept was introduced as \emph{action equivalence} in \cite{jveetal:2012} (and its precursors) and also subsequently used in that sense in \cite{kooirenne}. The definitions from \cite{jveetal:2012,kooirenne} do not deal very well with multi-pointed update modalities and with arbitrary update modalities, however, so we use a slightly adapted version.

We recall from the introduction that the updates $X$ we consider are relational model transformers and that such transformations are defined by pairwise relating pointed models: \begin{quote} {\em An update $X$ is a relation between pointed relational models.}\end{quote} In fact, three different things are called update: the {\em update relation} between pointed models, the {\em update modality} in a logical language, and, in some sense, the {\em update object}, often a name, that can be associated with the modality or the relation, such as an arrow update $(U,o)$. To simplify the presentation in this section we call the relations $X,Y,\dots$ and the modalities $[X],[Y],\dots$ and we do not consider the update objects separately: we identify them with the update relations.

A relation between pointed relational models can be a one-to-one relation, i.e., a function or a partial function, a one-to-finitely-many relation, and a one-to-infinitely-many relation. For example, it is a function for a pointed arrow update model, a partial function for a public announcement, a one-to-many relation for a multi-pointed arrow update model, and a one-to-infinitely-many relation for an arrow update quantifier. In the first place, one would now like to say that update relations $X$ and $Y$ are the same (are equivalent) if they define the same relation between pointed models, modulo bisimulation. In the second place, we also want to compare an update $X$ that is a partial function, i.e., with a restricted domain, to an update $Y$ that is a total function (or similarly for relations with restricted domains. In that kind of situation one would maybe like to say that updates $X$ and $Y$ are the same if $X$ and $Y$ define the same relation {\em on the domain of $X$}: we will then say that $X$ is conditionally equivalent to $Y$ (this relation is asymmetric). Such a requirement seems common practice in dynamic epistemic logic, and it is also respected in \cite{kooirenne}. We recall from Section \ref{sec.introduction} the `state eliminating' public announcement of $p$ (i) and the `arrow eliminating' public announcement of $p$ (ii), originating with \cite{gerbrandy:1999}: whenever $p$ can be truthfully announced, the pointed relational models resulting from executing (i) and (ii) are bisimilar, as in the example. But when $p$ is false, (ii) can be executed but not (i). So (i) and (ii) are equivalent on condition of the truth of the announcement. 

In view of these considerations, we propose the following definition. In the definition, for $X(M,s)$ read $\{ (M',s') \mid ((M,s),(M',s')) \in X \}$, and let $\dom(X)$ be $\{(M,s)\mid X(M,s)\not = \emptyset\}$, i.e., $\dom(X)$ is the domain of $X$ in the standard relational sense.\footnote{We recall that given a relational model $M = (S,R,V)$ or an arrow update model $U = (O,\aufunction)$, the sets $S = \Domain(M)$ resp.\ $O = \Domain(U)$ of objects on which the accessibility relations $R_a$ resp.\ arrow relations $\aufunction_a$ are defined are also called domains. This causes confusion if the arrow update model $U$ is considered as an update relation between sets of pointed relational models, with its domain $\dom(U)$ of application defined as above. As the term `domain' is extremely standard for both, we prefer to distinguish them by using different symbols $\Domain$ and $\dom$, instead of introducing a non-standard term for either one or the other.} Recall from Section~\ref{sec.structures} that two sets of pointed models are bisimilar if every pointed model in the first set is bisimilar to one in the second set and vice versa.

\begin{definition}[Update equivalence, update expressivity]
Given updates $X$ and $Y$, {\em $X$ is conditionally update equivalent to $Y$}, if for all $(M,s) \in \dom(X)$, $X(M,s) \bisim$ $Y(M,s)$. Further, {\em $X$ is update equivalent to $Y$}, if $X$ is conditionally update equivalent to $Y$, and $Y$ is conditionally update equivalent to $X$. Update modalities $[X]$ and $[Y]$ are (conditionally) update equivalent, if $X$ and $Y$ are (conditionally) update equivalent.
  
A language $\mathcal{L}$ is {\em at least as update expressive as} $\mathcal{L}'$ if for every update modality $[X]$ of $\mathcal{L}'$ there is an update modality $[Y]$ of $\mathcal{L}$ such that $X$ is conditionally update equivalent to $Y$; $\mathcal{L}$ is {\em equally update expressive as} $\mathcal{L}'$ (or `as update expressive as') if $\mathcal{L}$ is at least as update expressive as $\mathcal{L}'$ and $\mathcal{L}'$ is at least as update expressive as $\mathcal{L}$.
\end{definition}
We define `(strictly) more update expressive' and `incomparable in update expressivity' as usual. We also extend the usage of `update expressive' to the logics for the languages that we compare. Instead of `update equivalent' we may use `equivalent' if the context is clear. If updates $X$ and $Y$ are update equivalent, then $[X]\phi \eq [Y]\phi$ is valid. In the other direction, this is not always the case! In Section \ref{sec.rml} we give a counterexample.

We should stress that we do not claim that our definition is appropriate for all situations, merely that it gives an accurate view of the strengths of the different logics that we consider in this particular paper.

Let us now fill in the expressivity hierarchy for our target logics. The update expressivity of AUL is higher than that of PAL, and lower than that of AML \cite{kooietal:2011}. The comparison between AML and AUML that we will address in Section \ref{sec.arrowvaction} is less straightforward than that. In \cite{kooirenne} it was shown that AML  and AUML have the same update expressivity. That result does not distinguish between single-pointed and multi-pointed action models and arrow update models, however. Here, we therefore provide an alternative proof of their results that makes that distinction. Specifically, we show that the result from \cite{kooirenne} only applies to the multi-pointed case, but that single-pointed arrow update models are more update expressive than single-pointed action models.

To formulate such results we need to slightly expand our notation. We recall that in the language $\lang$ of AAUML (see Section~\ref{sec.syntax}) we permit multi-pointed arrow update modalities. If we only allow single-pointed arrow update modalities, let us call the set of validities AAUML$_1$. For the `logic' in the standard sense of the set of validities this makes no difference, so that out of the context of update expressivity we can continue to also let AAUML represent AAUML$_1$. Similarly, without the quantifiers, we distinguish AUML$_1$ from AUML, and we will later also introduce such a notational distinction for action model logics.

Adding quantification increases update expressivity, since the non-quantified logics cannot simulate a one-to-infinity relation. So, for example, APAL is more update expressive than PAL, and AAUML is more update expressive than AUML. (The distinction between single-pointed and multi-pointed is irrelevant here, as quantifying over all single-pointed updates is equivalent to quantifying over all multi-pointed updates.) When comparing the quantified logics among themselves, most pairs are incomparable. These incomparability results are all rather trivial, so we do not prove them here. The only comparable pair is (multi-pointed) AAUML and AAML, which have the same update expressivity because their underlying updates have the same update expressivity (Section \ref{sec.arrowvaction}). In Section \ref{sec.rml} we will show that RML is incomparable to the other quantified logics, and in particular that the AAUML and AAML quantifiers are contained in the RML quantifier.

The landscape of update expressivity is therefore as shown in Figure~\ref{fig:logics_update_expressivity}.

\begin{figure}
\center
\begin{tikzpicture}
\node (APAL) at (0,2) {APAL};
\node (AAUL) at (3,2) {AAUL};

\node (ML) at (-3,0) {ML};
\node (PAL) at (0,0) {PAL};
\node (AUL) at (3,0) {AUL};
\node (AMLs) at (0,-2) {AML$_1$};
\node (AUMLs) at (3,-2) {AUML$_1$};
\node (AMLm) at (0,-4) {AML};
\node (AUMLm) at (3,-4) {AUML};
\node (AAMLm) at (0,-6) {AAML};
\node (AAUMLm) at (3,-6) {AAUML};
\node (RML) at (-3,-2) {RML};

\draw[->] (PAL) -- (AUL);
\draw[<->] (AMLm) -- (AUMLm);

\draw[->] (PAL) -- (AMLs);
\draw[->] (AUL) -- (AUMLs);

\draw[->] (AMLs) -- (AMLm);
\draw[->] (AMLs) -- (AUMLs);
\draw[->] (AUMLs) -- (AUMLm);

\draw[->] (PAL) -- (APAL);
\draw[->] (AUL) -- (AAUL);
\draw[->] (ML) -- (RML);
\draw[->] (ML) -- (PAL);
\draw[->] (AMLm) -- (AAMLm);
\draw[->] (AUMLm) -- (AAUMLm);
\draw[<->] (AAMLm) -- (AAUMLm);
\end{tikzpicture}
\caption{The relative update expressivity of the update logics discussed in the paper. We assume transitive closure of arrows.}
\label{fig:logics_update_expressivity}
\end{figure}
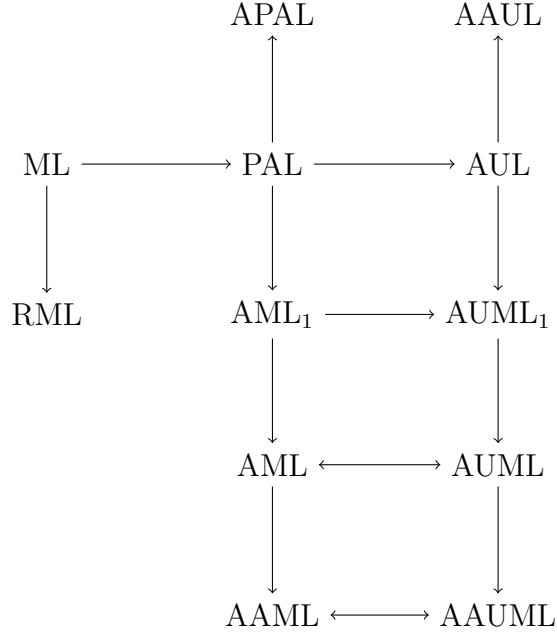

\section{Arrow updates versus action models} \label{sec.arrowvaction}

\subsection{Action model logic}
Arrow update model logic AUML is equally expressive as action model logic AML and their update expressivity relates in interesting ways. We build upon the results known from \cite{kooirenne} but our constructions and proofs are slightly different. First we need to define action models and their execution in relational models. An action model \cite{baltagetal:1998} is a structure like a relational model but with a precondition function instead of a valuation function. Executing an action model into a relational model means computing what is known as their restricted modal product. This product encodes the new state of information, after action execution. These are the technicalities:

An {\em action model} $E = (\Actions, \arel, \pre)$ consists of a {\em domain} $\Actions$ of {\em actions} $e,f,\dots$, an {\em accessibility function} $\arel: \Agents \imp {\mathcal P}(\Actions \times \Actions)$, where each $\arel_\agent$ is an accessibility relation, and a {\em precondition function} $\pre: \Actions \imp \lang$, where $\lang$ is a logical language.

Let additional to a pointed action model $(E,e)$ as above a pointed relational model $(M,\state)$ be given where $M = (\States, R, V)$. Let $M,\state \models \pre(\actiona)$. The update $(M \otimes E, (\state,e))$ is the pointed relational model where $M \otimes E = (\States', R', V')$ such that \[ \begin{array}{lcl}  \States' & = & \{ (\stateb,f) \mid M,\stateb \models \pre(f) \} \\ ((\stateb,f),(\stateb',f')) \in  R'_a & \text{iff} & (\stateb,\stateb') \in R_a \text{ and } (f,f') \in \arel_a \\ (\stateb,f) \in V'(\atom) & \text{iff} & \stateb \in V(\atom)
\end{array} \]
Action model modalities $[E,e]$ are interpreted similarly to arrow update modalities but unlike arrow update modalities are partial and not functional. Their execution depends on the truth of the precondition of the actual action (point) $e$ in the actual state $s$:
\[ M,s\models [E,e]\phi \text{ \ \ iff \ \ } M,s \models \pre(e) \text{ implies } M \otimes E, (s,e) \models \phi \]
Similarly to arrow update modalities we can conceive a modal logical language with $[E,e]\phi$ as an inductive language construct, for action models $E$ with finite domains. The logic is called {\em action model logic} AML. And also similarly we define multi-pointed action models by notational abbreviation, and informally consider such modalities as logical connectives binding formulas. 
As for AUML (see \cite{kooirenne} and Section \ref{sec.axiomatization}), there is a complete axiomatization, that is a rewrite system allowing to eliminate dynamic modalities \cite{baltagetal:1998,hvdetal.del:2007}. If we further extend the logical language with a quantifier $[\otimes]$ over action models, such that  
\[ M,s\models [\otimes]\phi \text{ \ \ iff \ \ } M,s \models [E,e]\phi \text{ for all action models } (E,e) \text{ satisfying } (*) \]
where (*) requires all preconditions of actions in $E$ to be in $\langml$, we get the language and logic of {\em arbitrary action model logic} AAML. Hales showed in \cite{hales2013arbitrary} that the (*) requirement can be relaxed, similarly to our Proposition \ref{prop:fully_arbitrary}. If such a distinction is necessary, the action model logics with only single-pointed action models are AML$_1$ and AAML$_1$. 
 
Example action models that are update equivalent to the example arrow update models of Section \ref{sec.aauml} are depicted in Figure~\ref{fig.actionmodelex}. We also depict their execution. The actions are given their preconditions as names. Note that the pointed relational model resulting from the (second) action of Anne privately learning that $p$ is bisimilar to the four-state model in Section \ref{sec.aauml} (Figure~\ref{fig.below}).

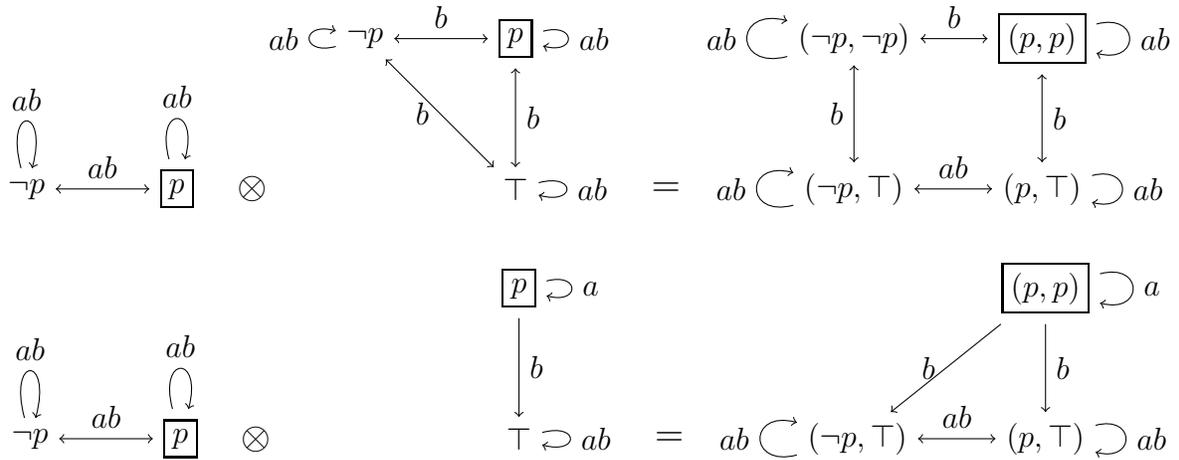
\begin{figure}[h]
\center
\begin{tikzpicture}
\node (0) at (0,0) {$\neg p$};
\node (1) at (2,0) {$\fbox{$p$}$};
\draw[<->] (0) -- node[above] {$ab$} (1);
\draw[->] (0) edge[loop above,looseness=15] node[above] {$ab$} (0); 
\draw[->] (1) edge[loop above,looseness=9] node[above] {$ab$} (1); 
\node (t) at (3,0) {\large $\otimes$};
\node (0r) at (6.5,0) {$\T$};
\node (1lr) at (4.5,2) {$\neg p$};
\node (1r) at (6.5,2) {$\fbox{$p$}$};
\draw[<->] (0r) -- node[right] {$b$} (1r);
\draw[<->] (0r) -- node[left] {$b$} (1lr);
\draw[<->] (1r) -- node[above] {$b$} (1lr);
\draw[->] (0r) edge[loop right,looseness=8] node[right] {$ab$} (0r); 
\draw[->] (1r) edge[loop right,looseness=6] node[right] {$ab$} (1r); 
\draw[->] (1lr) edge[loop left,looseness=6] node[left] {$ab$} (1lr); 
\node (t) at (8.5,0) {\large $=$};
\node (0r) at (11,0) {$(\neg p, \T)$};
\node (1r) at (13.5,0) {$(p, \T)$};
\node (0ra) at (11,2) {$(\neg p, \neg p)$};
\node (1ra) at (13.5,2) {$\fbox{$(p, p)$}$};
\draw[<->] (0r) -- node[above] {$ab$} (1r);
\draw[<->] (0ra) -- node[above] {$b$} (1ra);
\draw[<->] (0r) -- node[left] {$b$} (0ra);
\draw[<->] (1r) -- node[right] {$b$} (1ra);
\draw[->] (0r) edge[loop left,looseness=4] node[left] {$ab$} (0r); 
\draw[->] (1r) edge[loop right,looseness=4] node[right] {$ab$} (1r); 
\draw[->] (0ra) edge[loop left,looseness=4] node[left] {$ab$} (0ra); 
\draw[->] (1ra) edge[loop right,looseness=4] node[right] {$ab$} (1ra); 
\end{tikzpicture}

\bigskip

\noindent
\begin{tikzpicture}
\node (0) at (0,0) {$\neg p$};
\node (1) at (2,0) {$\fbox{$p$}$};
\draw[<->] (0) -- node[above] {$ab$} (1);
\draw[->] (0) edge[loop above,looseness=15] node[above] {$ab$} (0); 
\draw[->] (1) edge[loop above,looseness=9] node[above] {$ab$} (1); 
\node (t) at (3,0) {\large $\otimes$};
\node (0r) at (6.5,0) {$\T$};
\node (1r) at (6.5,2) {$\fbox{$p$}$};
\draw[<-] (0r) -- node[right] {$b$} (1r);
\draw[->] (0r) edge[loop right,looseness=8] node[right] {$ab$} (0r); 
\draw[->] (1r) edge[loop right,looseness=6] node[right] {$a$} (1r); 
\node (t) at (8.5,0) {\large $=$};
\node (0r) at (11,0) {$(\neg p, \T)$};
\node (1r) at (13.5,0) {$(p, \T)$};
\node (1ra) at (13.5,2) {$\fbox{$(p, p)$}$};
\draw[<->] (0r) -- node[above] {$ab$} (1r);
\draw[<-] (1r) -- node[right] {$b$} (1ra);
\draw[<-] (0r) -- node[left] {$b$} (1ra);
\draw[->] (0r) edge[loop left,looseness=4] node[left] {$ab$} (0r); 
\draw[->] (1r) edge[loop right,looseness=4] node[right] {$ab$} (1r); 
\draw[->] (1ra) edge[loop right,looseness=4] node[right] {$a$} (1ra); 
\end{tikzpicture}
\caption{Examples of action model execution}
\label{fig.actionmodelex}
\end{figure}

\subsection{From action models to arrow updates}
A given action model can be transformed into an arrow update model by decorating each arrow in the action model with a source condition that is the precondition of the source action and with a target condition that is the precondition of the target action. That is all. Technically:

Let $E = (\Actions,\arel,\pre)$ be given. Arrow update model $U(E) = (\austates,\aufunction)$ is defined as: $\austates = \Actions$, and for all agents $a$ and actions $e,e'$, $(e,\pre(e))\imp_\agent(e',\pre(e'))$ iff $(e,e') \in \arel_a$.\footnote{In \cite{kooirenne}, arrows $(e,\T)\imp_\agent(e',\phi')$ instead of $(e,\phi)\imp_\agent(e',\phi')$ are stipulated. Both constructions deliver the desired update equivalence.} We can now show that $(E,e)$ is update equivalent to $(U(E),e)$, on condition of the executability of the action $e$, i.e., restricted to the denotation of $\pre(e)$.

\begin{proposition}[\cite{kooirenne}] \label{prop.actiontoarrow} $(E,e)$ is conditionally update equivalent to $(\aumodel(E),e)$.
\end{proposition}
\begin{proof} Let $M = (\States,R,V)$, $M\otimes E = (S',R',V')$ and $M* U(E) = (S'',R'',V'')$.
Define the relation $\bisrel$ by $\bisrel = \{((s,e),(s,e))\in S'\times S''\mid (s,e)\in S'\}$. We will show that if $(M,s)\in\dom(E,e)$, so if $M,s\models \pre(e)$, then $\bisrel$ is a bisimulation linking $(s,e)\in S'$ to $(s,e)\in S''$.

Take any $((s_1,f_1),(s_1,f_1))\in\bisrel$. The {\bf atoms} clause is trivially satisfied as the states $s$ match (and updates do not change facts). We now consider {\bf forth}. Suppose $(s_1,f_1)R_a'(s_2,f_2)$. Given the definition of action model execution, this implies that $(s_1,s_2)\in R_a$, $(f_1,f_2)\in \arel_a$, $M,s_1\models \pre(f_1)$ and $M,s_2\models \pre(f_2)$. Because of how we constructed $U(E)$, it follows from $(f_1,f_2)\in \arel_a$ that $(f_1,\pre(f_1))\imp_a (f_2,\pre(f_2))$. Then from $M,s_1\models \pre(f_1)$ and $M,s_2\models \pre(f_2)$ it follows that $(s_1,f_1)R_a''(s_2,f_2)$. Also, $((s_2,f_2)(s_2,f_2)) \in \bisrel$ by the definition of the relation $\bisrel$. So {\bf forth} is satisfied. Finally, we consider {\bf back}. Suppose $(s_1,f_1)R_a''(s_2,f_2)$. Then we must have $(f_1,\pre(f_1))\imp_a (f_2,\pre(f_2))$ with $(f_1,f_2)\in \arel_a$, since $U(E)$ only contains arrows of that type. It follows that $M,s_1\models \pre(f_1)$ and $M,s_2\models \pre(f_2)$, so $(s_1,f_1)R_a'(s_2,f_2)$. Again, obviously, $((s_2,f_2)(s_2,f_2)) \in \bisrel$. The {\bf back} condition is therefore also satisfied.

We have now shown that $\bisrel$ is a bisimulation. Since $(s,e)\bisrel (s,e)$ this implies that $M\otimes E,(s,e)\bisim M*U(E),(s,e)$. This holds for every $(M,s)\in\dom(E,e)$, so $(E,e)$ is conditionally update equivalent to $(\aumodel(E),e)$.

\end{proof}

\begin{corollary} \label{cor.actiontoarrow} Let $F \subseteq \Domain(E)$. Then $(E,F)$ is conditionally update equivalent to $(\aumodel(E),F)$.
\end{corollary}

From Proposition \ref{prop.actiontoarrow} follows that, for all $\phi\in\langml$, $[E,e]\phi$ is equivalent to $\pre(e) \imp [\aumodel(E),e]\phi$. See also \cite{kooirenne}[Cor.\ 3.9]. This can be used as a clause in an inductively defined translation from the language of AUML to the language of AML. In Corollary \ref{cor.actiontoarrow} the condition for update equivalence is $\Vel_{e\in F} \pre(e)$.

The arrow update models constructed by the above procedure from the action model for Anne reading the letter containing $p$ while Bill may notice her doing so, and for Anne privately learning $p$, are as follows. Note that they are update equivalent (namely on their domain of execution) to the arrow update models for these actions presented in Section \ref{sec.aauml}. In the figure, by $\phi \ _{ab} \ \phi'$ we mean the two arrows  $\phi \imp_a \phi', \phi \imp_b \phi'$. On the left, the dual arrows for $b$ between outcomes have not been labelled. They are as expected: $\T \imp_b \neg p$, $p \imp_b \neg p$, $p \imp_b \T$. (They are update equivalent to the example arrow update models in Section \ref{sec.aaumlex} and Figure~\ref{fig.below}.)

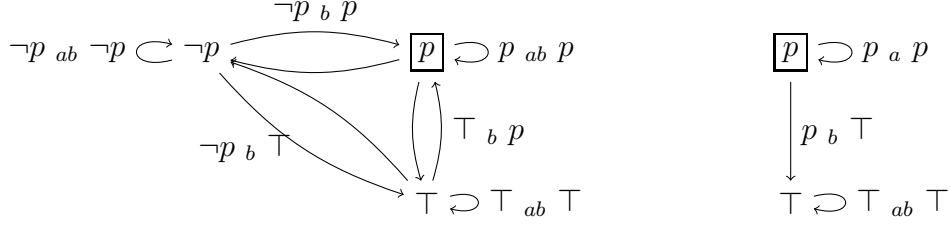
\begin{figure}[h]
\center
\begin{tikzpicture}
\node (0r) at (8,0) {$\T$};
\node (1lr) at (5,2) {$\neg p$};
\node (1r) at (8,2) {$\fbox{$p$}$};
\draw[->,bend right=15] (0r) to node[right] {$\T \ _b \ p$} (1r);
\draw[->,bend right=15] (0r) to (1lr); 
\draw[->,bend left=15] (1r) to (1lr); 
\draw[<-,bend left=15] (0r) to (1r); 
\draw[<-,bend left=15] (0r) to node[left] {$\neg p \ _b \ \T$} (1lr);
\draw[<-,bend right=15] (1r) to node[above] {$\neg p \ _b \ p$} (1lr);
\draw[->] (0r) edge[loop right,looseness=8] node[right] {$\T \ _{ab} \ \T$} (0r); 
\draw[->] (1r) edge[loop right,looseness=8] node[right] {$p \ _{ab} \ p$} (1r); 
\draw[->] (1lr) edge[loop left,looseness=8] node[left] {$\neg p \ _{ab} \ \neg p$} (1lr); 
\end{tikzpicture}
\hspace{2cm}
\begin{tikzpicture}
\node (0r) at (5,0) {$\T$};
\node (1r) at (5,2) {$\fbox{$p$}$};
\draw[<-] (0r) -- node[right] {$p \ _b \ \T$} (1r);
\draw[->] (0r) edge[loop right,looseness=9] node[right] {$\T \ _{ab} \ \T$} (0r); 
\draw[->] (1r) edge[loop right,looseness=8] node[right] {$p \ _a \ p$} (1r); 
\end{tikzpicture}
\caption{From action models to arrow updates --- example}
\label{fig.actiontoarrow}
\end{figure}

\subsection{From arrow updates to action models}

A given arrow update model can be transformed into an update equivalent action model by conditionalizing in each outcome over any possible `valuation' of (any subset of) the source and target conditions of all outcomes. This leads to an exponential blowup. (See \cite[Theorem 3.7]{kooirenne}. Our construction and subsequent proof are different.)  We proceed with the construction.

Let $U = (\austates,\aufunction)$ be given. Let $\Phi$ be the collection of all source and target conditions occurring in $U$: \[ \Phi = \{ \phi \mid \text{there are } a \in \Agents, \phi' \in \lang, o,o' \in \austates \text{ s.t.\ } (o,\phi)\imp_\agent(o',\phi') \text{ or } (o,\phi')\imp_\agent(o',\phi)\}.\] We consider the formulas in $\Phi$ as `atomic constituents' over which we consider `valuations' $v \in 2^\Phi$ (lower case, to distinguish it from the relational model valuation $V$, upper case). The characteristic formula of a valuation is $\delta_v := \Et_{\phi\in\Phi} \overline{\phi}$, where $\overline{\phi} = \phi$ if $v(\phi) = 1$ and $\overline{\phi} = \neg\phi$ if $v(\phi) = 0$. Action model $E(\aumodel) = (\Actions,\arel,\pre)$ is now such that: \[\begin{array}{lcl} \Actions &=& \austates \times 2^\Phi \\ ((o,v), (o',v')) \in \arel_a & \text{iff} & \is\phi,\phi' : (o,\phi)\imp_\agent(o',\phi'), v(\phi) = 1, v'(\phi') = 1 \\ \pre(o,v) &=& \delta_v \end{array} \]
Further, given $(U,o)$, its single point $o$ becomes a set of actions $E(o) := \{o\} \times 2^\Phi$. The corresponding action model $(E(U),E(o))$ is therefore multi-pointed (unless $\Phi = \{\T\}$ or $\Phi = \{\F\}$). We note that the preconditions of actions need not be consistent formulas, just as source and target conditions of arrows need not be consistent formulas. Our construction is therefore different from that in \cite{kooirenne}, wherein only $v \in 2^\Phi$ are considered for which $\delta_v$ is consistent. That construction is more economical, but computational efficiency is not our goal. We can now show that $(U,o)$ is update equivalent to $(E(U),E(o))$. Note that they both can be executed on the entire domain.

\begin{proposition} \label{prop.arrowtoaction} $(\aumodel,o)$ is update equivalent to $(E(\aumodel),E(o))$.
\end{proposition}
\begin{proof} Let $M=(S,R,V)$ and $s\in S$ be given, and let $M*U=(S',R',V')$ and $M\otimes E(U) = (S'',R'',V'')$. We show that for some $(o,v) \in E(o)$, $(M * U, (s,o)) \bisim (M \otimes E(U),(s,o,v))$. Define relation $\bisrel$ as follows: \[\begin{array}{lcl} (s,o)\bisrel (s,o,v) & \text{ iff } & M,s \models \delta_v\end{array} \] We show that $\bisrel$ is a bisimulation.

For {\bf forth}, assume $((s_1,o_1),(s_1,o_1,v_1)) \in \bisrel$ and $((s_1,o_1),(s_2,o_2)) \in R'_a$. The latter implies that there are $\phi_1,\phi_2\in\lang$ such that $(o_1,\phi_1) \imp_a (o_2,\phi_2)$, and that $M,s_1\models\phi_1$ and $M,s_2\models\phi_2$. Choose $v_1$ and $v_2$ such that $M,s_1\models \delta_{v_1}$ and $M,s_2 \models \delta_{v_2}$ (note that $v_1$ and $v_2$ exist and are unique). As $M,s_1\models \phi_1$ and $M,s_1\models \delta_{v_1}$ we have $v_1(\phi_1)=1$ and as $M,s_2\models\phi_2$ and $M,s_2\models \delta_{v_2}$ we have $v_2(\phi_2)=1$.
We claim that $(s_2,o_2,v_2)$ is the requested witness to close the {\bf forth} argument. Firstly, $((s_1,o_1,v_1),(s_2,o_2,v_2))\in R''_a$ because $(s_1,s_2) \in R_a$ and $((o_1,v_1),(o_2,v_2)) \in \arel_a$, where the latter follows from $(o_1,\phi_1) \imp_a (o_2,\phi_2)$, $v_1(\phi_1)=1$ and $v_2(\phi_2)=1$ (see the definition of $E(U)$ above). Secondly, $((s_2,o_2),(s_2,o_2,v_2) \in \bisrel$. Step {\bf back} is very similar; now use that $\pre(o_1,v_1) = \delta_{v_1}$ and that $\pre(o_2,v_2) = \delta_{v_2}$, and observe that $(E(U),E(o))$ can always be executed, it has precondition  $\Vel_{v \in 2^\Phi} \delta_v$, which is equivalent to $\T$.  
\end{proof}

\begin{corollary} \label{cor.arrowtoaction} Let now $Q \subseteq \Domain(\aumodel)$. Then $(\aumodel,Q)$ is update equivalent to $(E(\aumodel),E(Q))$.
\end{corollary}

As an example we now show the action models constructed by the above procedure from the two example arrow update models of Section \ref{sec.aauml}. The set of source and target condition formulas is $\{p,\neg p, \T\}$. Of their $8$ valuations the two non-trivial (and different) valuations are characterized by $p$ and $\neg p$. These formulas are also the action preconditions in the action points of the resulting action models. The reader may observe that these action models are, again, update equivalent to their `original' action models at the start of this section.

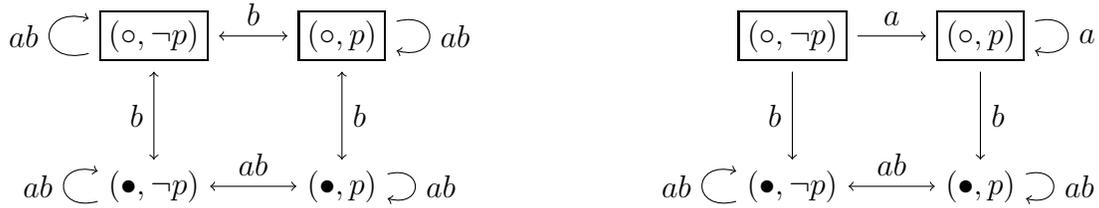
\begin{figure}[h]
\center
\begin{tikzpicture}
\node (0r) at (9.5,0) {$(\bullet,\neg p)$};
\node (1r) at (12,0) {$(\bullet, p)$};
\node (0ra) at (9.5,2) {$\fbox{$(\circ, \neg p)$}$};
\node (1ra) at (12,2) {$\fbox{$(\circ,p)$}$};
\draw[<->] (0r) -- node[above] {$ab$} (1r);
\draw[<->] (0ra) -- node[above] {$b$} (1ra);
\draw[<->] (0r) -- node[left] {$b$} (0ra);
\draw[<->] (1r) -- node[right] {$b$} (1ra);
\draw[->] (0r) edge[loop left,looseness=4] node[left] {$ab$} (0r); 
\draw[->] (1r) edge[loop right,looseness=4] node[right] {$ab$} (1r); 
\draw[->] (0ra) edge[loop left,looseness=4] node[left] {$ab$} (0ra); 
\draw[->] (1ra) edge[loop right,looseness=4] node[right] {$ab$} (1ra); 
\end{tikzpicture}
\hspace{2cm}
\begin{tikzpicture}
\node (0r) at (9.5,0) {$(\bullet,\neg p)$};
\node (1r) at (12,0) {$(\bullet,p)$};
\node (0ra) at (9.5,2) {$\fbox{$(\circ,\neg p)$}$};
\node (1ra) at (12,2) {$\fbox{$(\circ,p)$}$};
\draw[<->] (0r) -- node[above] {$ab$} (1r);
\draw[->] (0ra) -- node[above] {$a$} (1ra);
\draw[<-] (0r) -- node[left] {$b$} (0ra);
\draw[<-] (1r) -- node[right] {$b$} (1ra);
\draw[->] (0r) edge[loop left,looseness=4] node[left] {$ab$} (0r); 
\draw[->] (1r) edge[loop right,looseness=4] node[right] {$ab$} (1r); 
\draw[->] (1ra) edge[loop right,looseness=4] node[right] {$a$} (1ra); 
\end{tikzpicture}
\caption{From arrow updates to action models --- example}
\label{fig.arrowtoaction}
\end{figure}

\subsection{Relative update expressivity}

We are now prepared to harvest the update expressivity results. First, let us show that AUML$_1$ is more update expressive than AML$_1$.
\begin{proposition}
AUML$_1$ is more update expressive than AML$_1$.
\end{proposition}
\begin{proof}
From Proposition \ref{prop.actiontoarrow} follows that AUML$_1$ is at least as update expressive as AML$_1$. To show that the inclusion is strict, we need to show that for some $(U,o)$, there is no $(E,e)$ that induces the same relation.

Let $U$ be the arrow update model with a single outcome $o$, and a single arrow $(o,p)\rightarrow_a(o,\top)$. 
Suppose towards a contradiction that there is a single-pointed action model $(E,e)$ such that $(M*U,(s,o))\bisim(M\otimes E, (s,e))$ for every $(M,s)$. Then, in particular, $(E,e)$ must be executable everywhere, so $\pre(e)$ is equivalent to $\top$. Furthermore, if $M,s\models p \wedge \lozenge_a\top$, then $(M*U,(s,o))$ has at least one $a$-successor $(s',o)$ (note that $\Dia_a \top$ enforces that $s$ has an $a$-successor $s'$, and the arrow $s \imp_a s'$ satisfies $p$ at the source and $\top$ at the target). By assumption, $(M\otimes E,(s,e))$ is bisimilar to $(M*U,(s,o))$, so it has an $a$-successor $(s',e')$. This implies that $e\arel_a e'$. Now, consider a different model $(M',s)$ such that (i) $M',s\models \neg p$ and (ii) $s$ has an $a$-successor $s'$ such that $M',s'\models \pre(e')$. Then $(M'\otimes E,(s,e))$ has an $a$-successor, but $(M'*U,(s,o))$ has no $a$-successor, because $p$ is false in state $s$. This contradicts the assumption that $(M*U,(s,o))\bisim(M\otimes E, (s,e))$ for every $(M,s)$.
\end{proof}

\begin{proposition}
AUML is as update expressive as AML.
\end{proposition}
\begin{proof}
From Corollary \ref{cor.actiontoarrow} follows that AUML is at least as update expressive as AML. From Corollary \ref{cor.arrowtoaction} follows that AML is at least as update expressive as AUML.
\end{proof}
It is obvious that AML is more update expressive than AML$_1$, and that AUML is more update expressive than AUML$_1$.

\subsection{Applications illustrating the succinctness of arrow updates} \label{sec.applications}

In this section we give some application areas for the modelling of information change, where arrow updates are more succinct that corresponding (i.e., update equivalent) action models.

\paragraph*{Lying}
You are lying if you say that something is true while you believe that it is false --- and with the intention for the addressee(s) to believe that it is true. In the setting of public announcement logic a lie is a public announcement that is false. This is then contrasted to the (usual) public announcement that is true. Both are combined in the announcement that has no relation to its truth. This is known as the conscious update (\cite{gerbrandy:1999}, see the introduction) or, in a setting where lying is also distinguished, as the manipulative update \cite{hvdetal.lying:2012}. The arrow update for the conscious/manipulative update of $\phi$ is the singleton arrow update model with arrows \[ (o,\T) \imp_a (o,\phi) \] for all agents. Lying as such, wherein $\phi$ is required to be false, is not an arrow update as arrow update models have no preconditions, but such executability preconditions can be simulated as  antecedents of logical implications.

A problem with the manipulative update is that an agent who already believes the opposite of the lie, believes everything after incorporating the lie into her beliefs (believing a contradiction comes at that price). This is because the accessibility of that agent becomes empty as a result of the update. A solution to that is the {\em cautious update} that is also known as lying to {\em sceptical} agents \cite{steiner:2006,kooietal:2011,hvd.lying:2014}: the agent only updates her beliefs if the new information is consistent with her current beliefs. The arrow update for the {\em sceptical update of $\phi$} (again, we cannot model sceptical {\em lying} as this requires $\phi$ to be false) is a singleton arrow update model with arrows \[ \begin{array}{lcl} (o,\Dia_a \phi) &\imp_a& (o,\phi) \\ (o,\Box_a \neg\phi) &\imp_a& (o,\T) \end{array}\] for all agents \cite{kooietal:2011}.

Given a group of agents, some may believe the announcement, and others not. The arrow update modelling allows for this. However, if we make an action model for announcements to sceptical agents, we need to distinguish all combinations explicitly (we used a similar construction to get action model $E(U)$ from arrow update $U$, above, in order to prove Prop.~\ref{prop.arrowtoaction}). For example, for two agents $a$ and $b$, the action model consists of eight actions, with preconditions and accessibility relations as follows. In Figure \ref{fig.lying} we `name' the actions with their preconditions. To simplify the visualization, we do not label arrows with $a$ and $b$: solid arrows are for $a$ and dashed arrows for $b$. We also assume transitive closure of accessibility.

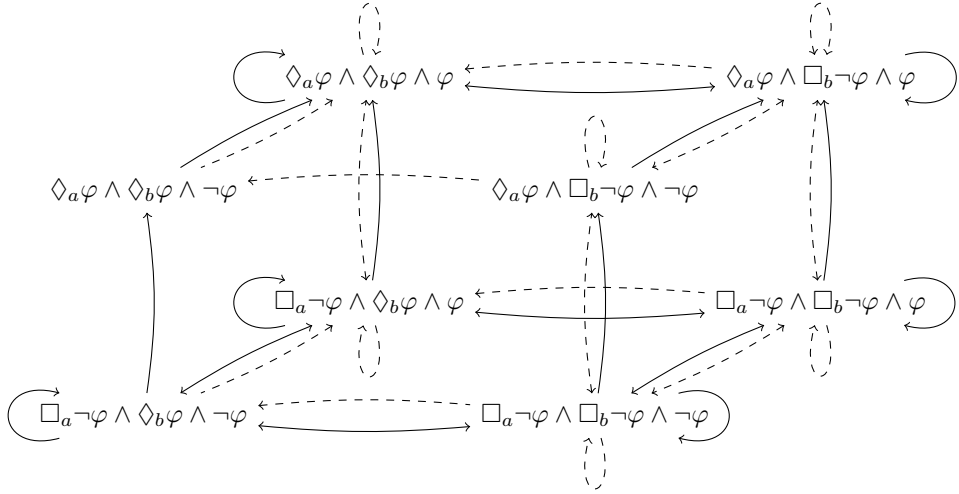
\begin{figure}[h]
\center
\begin{tikzpicture}
\node (01a) at (3,4.5) {\footnotesize $\Dia_a \phi \et \Dia_b \phi \et \phi$};
\node (11a) at (9,4.5) {\footnotesize $\Dia_a \phi \et \Box_b \neg\phi \et \phi$};
\node (00a) at (3,1.5) {\footnotesize $\Box_a \neg\phi \et \Dia_b \phi \et \phi$};
\node (10a) at (9,1.5) {\footnotesize $\Box_a \neg\phi \et \Box_b \neg \phi \et \phi$};
\node (01) at (0,3) {\footnotesize $\Dia_a \phi \et \Dia_b \phi \et \neg\phi$};
\node (11) at (6,3) {\footnotesize $\Dia_a \phi \et \Box_b \neg\phi \et \neg\phi$};
\node (00) at (0,0) {\footnotesize $\Box_a \neg\phi \et \Dia_b \phi \et \neg\phi$};
\node (10) at (6,0) {\footnotesize $\Box_a \neg\phi \et \Box_b \neg \phi \et \neg\phi$};
\draw[<->,bend right=5] (00) to (10);
\draw[->,bend right=8] (00) to (01);
\draw[->,bend right=8] (10) to (11);
\draw[->] (00) edge[loop left,looseness=4] (00); 
\draw[->] (10) edge[loop right,looseness=4] (10); 
\draw[<->,bend right=5] (00a) to (10a);
\draw[<->,bend right=5] (01a) to (11a);
\draw[->,bend right=8] (00a) to (01a);
\draw[->,bend right=8] (10a) to (11a);
\draw[->] (00a) edge[loop left,looseness=4] (00a); 
\draw[->] (10a) edge[loop right,looseness=4] (10a); 
\draw[->] (01a) edge[loop left,looseness=4] (01a); 
\draw[->] (11a) edge[loop right,looseness=4] (11a); 
\draw[<->,bend right=5] (00a) to (00);
\draw[<-,bend right=5] (01a) to (01);
\draw[<->,bend right=5] (10a) to (10);
\draw[<-,bend right=5] (11a) to (11);
\draw[<-,dashed,bend left=5] (00) to (10);
\draw[<-,dashed,bend left=5] (01) to (11);
\draw[<->,dashed,bend left=8] (10) to (11);
\draw[->,dashed] (10) edge[loop below,looseness=15] (10); 
\draw[->,dashed] (11) edge[loop above,looseness=15] (11); 
\draw[<-,dashed,bend left=5] (00a) to (10a);
\draw[<-,dashed,bend left=5] (01a) to (11a);
\draw[<->,dashed,bend left=8] (00a) to (01a);
\draw[<->,dashed,bend left=8] (10a) to (11a);
\draw[->,dashed] (00a) edge[loop below,looseness=15] (00a); 
\draw[->,dashed] (10a) edge[loop below,looseness=15] (10a); 
\draw[->,dashed] (01a) edge[loop above,looseness=15] (01a); 
\draw[<->,dashed] (11a) edge[loop above,looseness=15] (11a); 
\draw[<-,dashed,bend left=5] (00a) to (00);
\draw[<-,dashed,bend left=5] (01a) to (01);
\draw[<->,dashed,bend left=5] (10a) to (10);
\draw[<->,dashed,bend left=5] (11a) to (11);

\end{tikzpicture}
\caption{Action model for lying}
\label{fig.lying}
\end{figure}

\noindent For $n$ agents there are $O(2^n)$ actions in the action model. In \cite{kooietal:2011} this example is treated in greater detail, and also other, similar, examples are shown for which arrow updates are shown to be more succinct (exponentially smaller).

\paragraph*{Attentive announcements}
Another example where action models are exponentially bigger than arrow updates is that of the {\em attention-based announcements} of \cite{bolanderetal:2016}. This work presents a logic of announcements that are only `heard' (received) by agents paying attention to it, paying attention to the announcer, so to speak. Such announcements are modelled employing an auxiliary set of designated `attention (propositional) variables' $h_a$ expressing that agent $a$ pays attention. The corresponding arrow update model has domain $\{o,o'\}$, both outcomes designated, and with arrows \[ \begin{array}{lcl} (o,h_a) &\imp_a& (o,\phi) \\ (o,\neg h_a) &\imp_a& (o',\T) \\ (o',\T) &\imp_a& (o',\T) \end{array}\] for all agents. It cannot be modelled with a (singleton) \cite{kooietal:2011} arrow update, the resulting relational model is typically larger than then model before the update, as the agents not paying attention believe that no announcement was made, and thus reason about the structure of the entire initial model. Incorporating the announcements depends on $h_a$ being true or false, just as for sceptical announcements it depends on $\Dia_a \phi$ being true or false. So, this is similar. But not entirely so, because an agent not paying attention is, so to speak, inconscious of the announcement, and thus believes that the (entire) original model still encodes her beliefs), whereas an agent believing the opposite of the announcement `knows' that if she where to have found the announcement believable, she would have changed her beliefs. So these are different parts of the same model, it is a mere restriction of the accessibility relation. Again, for attentive announcements, a corresponding action model is of exponential size, as any subset of agents may or may not be paying attention. See \cite{bolanderetal:2016}.

\paragraph*{Comparative size of action models and arrow updates}
In general, if the observational powers of {\em all} agents are commonly known to be partial, then we can expect arrow updates for such dynamic phenomena to be exponentially smaller than corresponding action models. This was the case for announcements to sceptical lying and for attention-based announcements, and also for: agents making broadcasts (to all agents), agents seeing each other depending on their orientation, partial networks representing agents with neighbours or friends, etc. On the other hand, dynamic phenomena where all agents observe (some, few) designated agents have similarly-sized arrow updates and action models, such as: the private announcement to an individual agent or a subgroup of agents, and gossip scenarios where two agents call each other in order to exchange secrets, and where this call may be partially observed by all other agents. We do not know of scenarios where action models are more succinct than arrow updates.

\section{Arbitrary arrow updates versus refinements} \label{sec.rml} \label{sec.arrowvrefine}

\subsection{Refinement modal logic}

We now compare the arbitrary arrow update modality of AAUML to the {\em refinement} quantifier of refinement modal logic RML \cite{bozzellietal.inf:2014}. Let us first be precise about its syntax and semantics. 

We recall the definition of bisimulation in Section \ref{sec.structures}. If {\bf atoms} and {\bf back} hold, we call the relation a {\em refinement} (and dually, if {\bf atoms} and {\bf forth} hold, we call the relation a {\em simulation}). In  \cite{bozzellietal.inf:2014} such a refinement relation is considered for any subset of the set of agents and defined as follows: 

A relation ${\mathfrak R}_\Group$ that satisfies {\bf atoms},  {\bf back-$\agent$}, and {\bf forth-$\agent$} for every $\agent\in\Agents\setminus\Group$, and that satisfies {\bf atoms}, and {\bf back-$\agentb$} for every $\agentb\in\Group$, is a $\Group$-{\em refinement}, we say that $(M',{\state'})$ {\em refines} $(M,\state)$ for group of agents $\Group$, and we write $(M,\state) \lumis_\Group (M',{\state'})$. An $\Agents$-{\em refinement} we call a {\em refinement} (clearly any $\Group$-refinement is also an $\Agents$-refinement and thus a `refinement' plain and simple), and $(M,\state) \lumis_\Agents (M',{\state'})$ is denoted $(M,\state) \lumis (M',{\state'})$. With this relation comes a corresponding modality in the obvious way. Let $B \subseteq A$, and $(M,s)$ and $\phi$ given, then \[ M,s \models [\lumis]_B \phi \text{ \ \ \ iff \ \ \ } M',s' \models \phi \text{ for all } (M',s') \text{ such that } (M,s) \lumis_B (M',s'). \]

We continue with the comparison of the refinement quantifier to the other quantifiers. 
We first focus on the refinement relation $\lumis$ for the set of all agents, to which corresponds the $[\lumis]$ modality. Consider three different ways to define quantification in information changing modal logics. We formulate them suggestively so that their correspondences stand out, where we recall Proposition \ref{prop:fully_arbitrary} that the restrictions on source and target conditions need not be met when interpreting $[\AAUL]$, and similarly, \cite{hales2013arbitrary} showed that the restrictions on action preconditions need not be met when interpreting $[\otimes]$.

\[ \begin{array}{lcl} 
M,s \models [\AAUL] \phi & \text{iff} & M,s \models [U,o] \phi \text{ for all arrow update models } (U,o) \\ 
M,s \models [\lumis] \phi & \text{iff} & M',s' \models \phi \text{ for all refinements } (M',s') \\ 
M,s \models [\otimes] \phi & \text{iff} & M',s' \models [E,e] \phi \text{ for all action models } (E,e) 
\end{array} \]

\begin{theorem} \label{theo.drie}
Let $\phi \in \langml$. Then $[\AAUL] \phi$, $[\lumis]\phi$, $[\otimes]\phi$ are pairwise equivalent.
\end{theorem}
\begin{proof} \ \begin{itemize} \item The equivalence of $[\lumis]\phi$ to $[\otimes]\phi$ was shown in \cite{hales2013arbitrary}.
\item To show that $[\otimes]\phi$ is equivalent to $[\AAUL]\phi$, we use the semantics of these modalities. Let us do this for the diamond version. Both directions of the equivalence need to be shown. 

First, suppose that

{\centering $M,s \models \dia{\otimes}\phi$.

}

According to the semantics of $\dia{\otimes}$ this is equivalent to 

{\centering $\is (E,e): M,s \models \dia{E,e}\phi$.

}

By Proposition~\ref{prop.actiontoarrow} that, in turn, is equivalent to

{\centering $\is (U(E),e): M,s \models \pre(e) \ \text{and} \ M,s \models \dia{U(E), e} \phi$.

}

This implies that, in particular,

{\centering $\is (U(E),e): M,s \models \dia{U(E), e} \phi$

}
and therefore, by the semantics of $\dia{\AAUL}$, that

{\centering $M,s \models \dia{\AAUL}\phi$.

}
\medskip
For the other direction, suppose that 

{\centering $M,s \models \dia{\AAUL}\phi$.

}
According to the semantics of $\dia{\AAUL}$ this is equivalent to

{\centering $\is (U,o): M,s \models \dia{U,o}\phi$.

}
By Proposition~\ref{prop.arrowtoaction} that, in turn, is equivalent to

{\centering $\is (E(U),E(o)): M,s \models \dia{E(U), E(o)} \phi$.

}
In particular, this implies that for some $e\in E(o)$ we have

{\centering $\is (E(U),e): M,s \models \dia{E(U), e} \phi$

}
and therefore, by the semantics of $\dia{\otimes}$, that

{\centering $M,s \models \dia{\otimes}\phi$.

}

%
%
%
%
%

\item
To show that $[\AAUL] \phi$ is equivalent to $[\lumis]\phi$ we use the previous two equivalences.
\end{itemize} \vspace{-.7cm}
\end{proof}

The theorem is formulated to make the correspondence between the three quantifiers stand out. Alternatively, we can have an inductively defined translation between the language $\lang$ (of AAUML) and the language of arbitrary action model logic AAML that is compositional to the extent that arrow update quantifiers are translated into action model quantifiers (Theorem \ref{theo.drie}) and arrow update models into action models (Proposition \ref{prop.arrowtoaction}), and vice versa (Proposition \ref{prop.actiontoarrow}).

\subsection{Update expressivity}

Considering that $[\otimes]\phi$, $[\AAUL]\phi$ and $[\lumis]\phi$ are equivalent for  basic modal formulas $\phi$ (Theorem~\ref{theo.drie}), and that $[\otimes]$ and $[\AAUL]$ have the same update expressivity (Section \ref{sec.updateex}), one might expect all three logics  AAML, AAUML, and RML to have the same update expressivity. This, however, is not so, because $[\otimes]$ and $[\AAUL]$ are finitary quantifiers --- they quantify over, respectively, {\em finite} action models and over {\em finite} arrow update models --- whereas refinements can be infinitary.

For one example, consider the relational model $N$ consisting of all valuations, with the universal relation on that domain for all agents, and any state $t$ in that domain. Clearly, the restriction of $N$ to the singleton model consisting of $t$ (wherein the agents have common knowledge of the valuation in $t$) is a refinement of $(N,t)$. It can be obtained by successively announcing the value of each of the infinite number of atoms. However, it cannot be obtained by a single announcement (or, equivalently, by any finite sequence of those).

For another example, consider the models $M$, with as single state $s_0$, and $M'$, with $s_0$ as its leftmost state, in Figure~\ref{fig.modelsmandm}. The pointed model $(M',s_0)$ is a refinement of $(M,s_0)$. But $M'$ contains infinitely many states that are not bisimilar to one another. Furthermore, every arrow update model $U$ in the logical language is finite, so every product of $U$ with $M$ is finite (and therefore contains finitely many non-bisimilar states). As a result, there is no $(U,o)$ such that $(M',s_0)$ is bisimilar to $(M*U,(s,o))$.

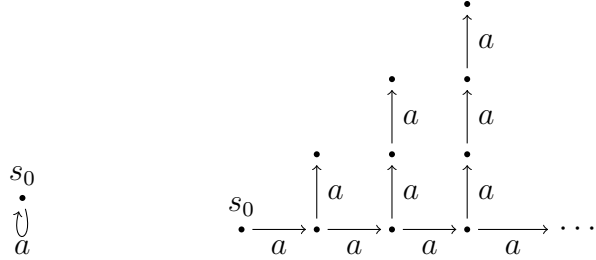
\begin{figure}[h]
\center
\begin{tikzpicture}
	\fill (0,0) circle (0.04) node (o0) {};
	\draw (o0) edge[loop below] (o0) node[midway,below=12pt] {$a$};
	\draw (0,0.3) node {$s_0$};
\end{tikzpicture}
\hspace{2cm}
\begin{tikzpicture}
	\fill (0,0) circle (0.04) node (o0) {};

	\fill (1,0) circle (0.04) node (o10) {};
	\fill (1,1) circle (0.04) node (o11) {};

	\fill (2,0) circle (0.04) node (o20) {};
	\fill (2,1) circle (0.04) node (o21) {};
	\fill (2,2) circle (0.04) node (o22) {};
	
	\fill (3,0) circle (0.04) node (o30) {};
	\fill (3,1) circle (0.04) node (o31) {};
	\fill (3,2) circle (0.04) node (o32) {};
	\fill (3,3) circle (0.04) node (o33) {};
	
	\node (o4) at (4.5,0) {$\cdots$};
	
	\draw[->] (o0) -- (o10) node[midway,below] {$a$};
	\draw[->] (o10) -- (o20) node[midway,below] {$a$};
	\draw[->] (o20) -- (o30) node[midway,below] {$a$};
	\draw[->] (o30) -- (o4) node[midway,below] {$a$};

	\draw[->] (o10) -- (o11) node[midway,right] {$a$};
	
	\draw[->] (o20) -- (o21) node[midway,right] {$a$};
	\draw[->] (o21) -- (o22) node[midway,right] {$a$};
	
	\draw[->] (o30) -- (o31) node[midway,right] {$a$};
	\draw[->] (o31) -- (o32) node[midway,right] {$a$};
	\draw[->] (o32) -- (o33) node[midway,right] {$a$};
	\draw (0,0.3) node {$s_0$};
\end{tikzpicture}
\caption{Arbitrary arrow updates and refinements are incomparable}
\label{fig.modelsmandm}
\end{figure}

It follows that arbitrary arrow updates are \emph{not} at least as update expressive as refinements. But it also follows that refinements are not at least as update expressive as arbitrary arrow updates, since you cannot choose to exclude the above model $(M',s_0)$ when performing a refinement in $(M,s_0)$.
\begin{proposition}
RML and AAUML are incomparable in update expressivity.
\end{proposition}
And therefore RML and AAML are also incomparable. The reason that $[\AAUL]\phi$ and $[\lumis]\phi$ (and $[\otimes]\phi$) are nonetheless equivalent is that while there is no $(U,o)$ such that $(M*U,(s_0,o))$ is bisimilar to $(M',s_0)$, it is the case that for every $n$ there is an $(U_n,o_n)$ such that $(M*U_n,(s_0,o_n))$ is $n$-bisimilar to $(M',s_0)$. Since every formula in the languages under consideration is of finite depth, such finite approximations of $M'$ suffice.

Finally, we should note that the incomparability already applies to the language for RML with only the $[\lumis]$ modality. The language above, as in \cite{bozzellietal.inf:2014}, has $[\lumis]_B$ modalities for any subgroup $B \subseteq \Agents$, meaning that, modulo bisimulation, only arrows in $B$ are removed from a relational model. Similarly to the argument above it follows that this would only further increase expressivity.

\section{Conclusions and further research} \label{sec.conclusion}

\paragraph*{Conclusions} We presented {\em arbitrary arrow update model logic} (AAUML). We provided an axiomatization of AAUML, which also demonstrates that AAUML is decidable and equally expressive as multi-agent modal logic. We established arrow update model synthesis for AAUML. We determined the update expressivity hierarchy including AAUML and many other update logics, including other arrow update logics, action model logics, and refinement modal logic. 

\paragraph*{Further research on $B$-restricted arrow update synthesis} Let $B$ be any subset of the set of all agents. Building upon the $B$-refinements of \cite{bozzellietal.inf:2014} and motivated by a similar approach used in \cite{hales2013arbitrary}, a variant of the synthesis problem for AAUML is to consider \emph{$B$-restricted} arrow update models. Roughly speaking, a $B$-restricted arrow update model represents an event where only the agents in $B$ can gain more factual information, while the agents outside $B$ remain at least as uncertain as they were before the event. The $B$-restricted synthesis problem can be solved in a very similar way to the unrestricted problem that we presented in this paper. 

Similarly to how arrow update models have the same update expressivity as action models and refinements, $B$-restricted arrow update models have the same update expressivity as $B$-restricted action models, and $B$-refinements have larger update expressivity.

Formally introducing $B$-restricted arrow update models, and showing that the results apply there as well, would require a lot of complicate notation and several complex definitions. So for the sake of readability we did not include them in this paper.

\paragraph*{Knowledge and belief} Arrow updates result in changes of knowledge and belief. Of particular interest are therefore updates that preserve the $\mathcal{S}5$ or $\mathcal{KD}45$ nature of relational models. It is unclear how to enforce such preservation semantically, as discussed in Section \ref{sec.aaumlex}. Such issues need to be resolved in order to find an axiomatization for arrow update logic, or arrow update model logic, or quantified versions of these, for the class $\mathcal{S}5$ or $\mathcal{KD}45$. These well-known problems  \cite{kooietal:2011,hvdetal.aij:2017} are related to similar issues for refinement modal logic, as refinement also is relational change. Non-trivial $\mathcal{S}5$ and $\mathcal{KD}45$ versions of refinement modal logic have been proposed in work by Hales {\em et al.} \cite{halesetal:2011,hales.aiml:2012}, and particular mention deserves  Hales' Ph.D.\ thesis \cite[Section 9.3, Section 9.4]{hales:2016} wherein the axiomatization AAUML is adapted to the classes $\mathcal{KD}45$ and $\mathcal{S}5$, respectively. These diverse results might inspire similar solutions for truly `epistemic' arrow update logics.

\paragraph*{Work in progress on the complexity of synthesis} We have shown that it is possible to perform synthesis for AAUML, and described an algorithm that does this synthesis. We have not, however, discussed the {\em computational complexity} of that algorithm. The complexity is non-elementary. The non-elementary blowup occurs when $\langle \AAUL\rangle$ operators are nested. In general, if $\phi$ is a  basic modal formula, then the procedure outlined in this paper allows us to find a basic modal formula $\phi'$ that is equivalent to $\langle \AAUL\rangle \phi$, but the size of this $\phi'$ is, in the worst case, exponential in the size of $\phi$. As a result, a formula $\psi$ containing $n$ nested $\langle \AAUL\rangle$ operators can be translated to an equivalent formula $\psi'$ of modal logic using this method, but both the formula $\psi'$ and the arrow update model synthesized for the outermost $\langle \AAUL\rangle$ operator will have experienced exponential blowup $n$ times, resulting in a non-elementary blowup overall. We suspect that this is unavoidable, i.e., that the difficulty of the synthesis problem is non-elementary. We do not, for now, have a hardness proof, however. 

\bibliographystyle{plain}
\bibliography{biblio2019}

\providecommand{\noopsort}[1]{}
\begin{thebibliography}{10}

\bibitem{arecesetal:2012}
C.~Areces, R.~Fervari, and G.~Hoffmann.
\newblock Moving arrows and four model checking results.
\newblock In {\em Proc.\ of 19th WoLLIC}, pages 142--153. Springer, 2012.
\newblock LNCS 7456.

\bibitem{aucher:2010}
G.~Aucher.
\newblock Characterizing updates in dynamic epistemic logic.
\newblock In {\em Proceedings of Twelfth KR}. AAAI Press, 2010.

\bibitem{aucher.jancl:2011}
G.~Aucher.
\newblock {DEL}-sequents for progression.
\newblock {\em Journal of Applied Non-Classical Logics}, 21(3-4):289--321,
  2011.

\bibitem{aucheretal:2009}
G.~Aucher, P.~Balbiani, L.~{Fari{\~n}as del Cerro}, and A.~Herzig.
\newblock Global and local graph modifiers.
\newblock {\em Electr. Notes Theor. Comput. Sci.}, 231:293--307, 2009.

\bibitem{balbianietal:2008}
P.~Balbiani, A.~Baltag, H.~van Ditmarsch, A.~Herzig, T.~Hoshi, and T.~De Lima.
\newblock `{K}nowable' as `known after an announcement'.
\newblock {\em Review of Symbolic Logic}, 1(3):305--334, 2008.

\bibitem{baltagetal:1998}
A.~Baltag, L.S. Moss, and S.~Solecki.
\newblock The logic of public announcements, common knowledge, and private
  suspicions.
\newblock In {\em Proc.\ of 7th TARK}, pages 43--56. Morgan Kaufmann, 1998.

\bibitem{modalhandbook}
P.~Blackburn, J.~van Benthem, and F.~Wolter, editors.
\newblock {\em Handbook of Modal Logic}.
\newblock Elsevier, 2006.

\bibitem{bolanderetal:2016}
T.~Bolander, H.~van Ditmarsch, A.~Herzig, E.~Lorini, P.~Pardo, and
  F.~Schwarzentruber.
\newblock Announcements to attentive agents.
\newblock {\em Journal of Logic, Language and Information}, 25(1):1--35, 2016.

\bibitem{bozzellietal.inf:2014}
L.~Bozzelli, H.~van Ditmarsch, T.~French, J.~Hales, and S.~Pinchinat.
\newblock Refinement modal logic.
\newblock {\em Information and Computation}, 239:303--339, 2014.

\bibitem{frenchetal:2014}
T.~French, J.~Hales, and E.~Tay.
\newblock A composable language for action models.
\newblock In R.~Gor{\'{e}}, B.P. Kooi, and A.~Kurucz, editors, {\em Advances in
  Modal Logic 10}, pages 197--216. College Publications, 2014.

\bibitem{frenchetal:2008}
T.~French and H.~van Ditmarsch.
\newblock Undecidability for arbitrary public announcement logic.
\newblock In {\em Advances in Modal Logic 7}, pages 23--42, London, 2008.
  College Publications.

\bibitem{gerbrandy:1999}
J.D. Gerbrandy.
\newblock {\em Bisimulations on Planet Kripke}.
\newblock PhD thesis, University of Amsterdam, 1999.
\newblock ILLC Dissertation Series DS-1999-01.

\bibitem{gerbrandyetal:1997}
J.D. Gerbrandy and W.~Groeneveld.
\newblock Reasoning about information change.
\newblock {\em Journal of Logic, Language, and Information}, 6:147--169, 1997.

\bibitem{hales2013arbitrary}
J.~Hales.
\newblock Arbitrary action model logic and action model synthesis.
\newblock In {\em Proc.\ of 28th {LICS}}, pages 253--262. IEEE, 2013.

\bibitem{hales:2016}
J.~Hales.
\newblock {\em Quantifying over epistemic updates}.
\newblock PhD thesis, School of Computer Science \& Software Engineering,
  University of Western Australia, 2016.
\newblock
  https://research-repository.uwa.edu.au/en/publications/quantifying\--over\--epistemic\--updates.

\bibitem{halesetal:2011}
J.~Hales, T.~French, and R.~Davies.
\newblock Refinement quantified logics of knowledge.
\newblock {\em Electr. Notes Theor. Comput. Sci.}, 278:85--98, 2011.

\bibitem{hales.aiml:2012}
J.~Hales, T.~French, and R.~Davies.
\newblock Refinement quantified logics of knowledge and belief for multiple
  agents.
\newblock In {\em Advances in Modal Logic 9}, pages 317--338. College
  Publications, 2012.

\bibitem{kooi:2003}
B.~Kooi.
\newblock {\em Knowledge, Chance, and Change}.
\newblock PhD thesis, University of Groningen, 2003.
\newblock ILLC Dissertation Series DS-2003-01.

\bibitem{kooi.jancl:2007}
B.~Kooi.
\newblock Expressivity and completeness for public update logics via reduction
  axioms.
\newblock {\em Journal of Applied Non-Classical Logics}, 17(2):231--254, 2007.

\bibitem{kooietal:2011}
B.~Kooi and B.~Renne.
\newblock Arrow update logic.
\newblock {\em {R}eview of {S}ymbolic {L}ogic}, 4(4):536--559, 2011.

\bibitem{kooirenne}
B.~Kooi and B.~Renne.
\newblock Generalized arrow update logic.
\newblock In {\em Proc.\ of 13th {TARK}}, pages 205--211, 2011.
\newblock Poster presentation.

\bibitem{plaza:1989}
J.A. Plaza.
\newblock Logics of public communications.
\newblock In {\em Proc.\ of the 4th ISMIS}, pages 201--216. Oak Ridge National
  Laboratory, 1989.

\bibitem{steiner:2006}
D.~Steiner.
\newblock A system for consistency preserving belief change.
\newblock In {\em Proc.\ of the ESSLLI Workshop on Rationality and Knowledge},
  pages 133--144, 2006.

\bibitem{jfak.sabotage:2005}
J.~van Benthem.
\newblock An essay on sabotage and obstruction.
\newblock In {\em Mechanizing Mathematical Reasoning}, volume 2605 of {\em LNCS
  2605}, pages 268--276. Springer, 2005.

\bibitem{jfaketal.lcc:2006}
J.~van Benthem, J.~van Eijck, and B.~Kooi.
\newblock Logics of communication and change.
\newblock {\em Information and Computation}, 204(11):1620--1662, 2006.

\bibitem{hvd.thesis:2000}
H.~van Ditmarsch.
\newblock {\em Knowledge games}.
\newblock PhD thesis, University of Groningen, 2000.
\newblock ILLC Dissertation Series DS-2000-06.

\bibitem{hvd.lying:2014}
H.~van Ditmarsch.
\newblock Dynamics of lying.
\newblock {\em Synthese}, 191(5):745--777, 2014.

\bibitem{hvdetal.handbook:2015}
H.~van Ditmarsch, J.Y. Halpern, W.~van~der Hoek, and B.~Kooi, editors.
\newblock {\em Handbook of epistemic logic}. College Publications, 2015.

\bibitem{hvdetal.world:2008}
H.~van Ditmarsch and B.~Kooi.
\newblock Semantic results for ontic and epistemic change.
\newblock In {\em Proc.\ of 7th LOFT}, Texts in Logic and Games 3, pages
  87--117. Amsterdam University Press, 2008.

\bibitem{hvdetal.aamas:2005}
H.~van Ditmarsch, W.~van~der Hoek, and B.~Kooi.
\newblock Dynamic epistemic logic with assignment.
\newblock In {\em Proc.\ of 4th AAMAS}, pages 141--148. ACM, 2005.

\bibitem{hvdetal.del:2007}
H.~van Ditmarsch, W.~van~der Hoek, and B.~Kooi.
\newblock {\em Dynamic Epistemic Logic}, volume 337 of {\em Synthese Library}.
\newblock Springer, 2008.

\bibitem{hvdetal.aij:2017}
H.~van Ditmarsch, W.~van~der Hoek, B.~Kooi, and L.B. Kuijer.
\newblock Arbitrary arrow update logic.
\newblock {\em Artif. Intell.}, 242:80--106, 2017.

\bibitem{hvdetal.undecidable:2017}
H.~van Ditmarsch, W.~van~der Hoek, and L.B. Kuijer.
\newblock The undecidability of arbitrary arrow update logic.
\newblock {\em Theor. Comput. Sci.}, 693:1--12, 2017.

\bibitem{hvdetal.lying:2012}
H.~van Ditmarsch, J.~van Eijck, F.~Sietsma, and Y.~Wang.
\newblock On the logic of lying.
\newblock In {\em Games, Actions and Social Software}, LNCS 7010, pages 41--72.
  Springer, 2012.

\bibitem{jveetal:2012}
J.~van Eijck, J.~Ruan, and T.~Sadzik.
\newblock Action emulation.
\newblock {\em Synthese}, 185(1):131--151, 2012.

\bibitem{EijckSW11}
J.~van Eijck, F.~Sietsma, and Y.~Wang.
\newblock Composing models.
\newblock {\em Journal of Applied Non-Classical Logics}, 21(3-4):397--425,
  2011.

\end{thebibliography}

\end{document}